\documentclass[pra, singlecolumn, superscriptaddress]{revtex4}
\usepackage[latin9]{inputenc}
\usepackage{amsmath}
\usepackage{amssymb, enumerate}
\usepackage{color}
\usepackage{babel}
\usepackage{graphicx}
\usepackage{epstopdf}
\usepackage{float}
\usepackage{thmtools}
\usepackage{thm-restate}
\usepackage{hyperref}
\usepackage{cleveref}
\usepackage{enumerate}
\usepackage{lipsum}
\usepackage{fixmath}
\usepackage{dsfont}
\makeatletter

\@ifundefined{textcolor}{}
{%
 \definecolor{BLACK}{gray}{0}
 \definecolor{WHITE}{gray}{1}
 \definecolor{RED}{rgb}{1,0,0}
 \definecolor{GREEN}{rgb}{0,1,0}
 \definecolor{BLUE}{rgb}{0,0,1}
 \definecolor{CYAN}{cmyk}{1,0,0,0}
 \definecolor{MAGENTA}{cmyk}{0,1,0,0}
 \definecolor{YELLOW}{cmyk}{0,0,1,0}
}

\newtheorem{theorem}{Theorem}

\newtheorem{proposition}[theorem]{Proposition}
\newtheorem{example}[theorem]{Example}

\newtheorem{lem}[theorem]{Lemma}
\newtheorem{definition}[theorem]{Definition}
\newtheorem{fact}[theorem]{Fact}

\newenvironment{proof}[1][Proof]{\noindent\textit{#1.} }{\ \rule{0.5em}{0.5em}}
\newenvironment{sproof}[1][Sketch of proof]{\noindent\textit{#1.} }{\ \rule{0.5em}{0.5em}}

\makeatother

\begin{document}

\title{Single trusted qubit is necessary and sufficient for quantum realisation of extremal no-signaling correlations}
\author{Ravishankar Ramanathan}
\affiliation{Department of Computer Science, The University of Hong Kong, Pokfulam Road, Hong Kong}
\email{ravi@cs.hku.hk}
\author{Micha{\l} Banacki}
\affiliation{International Centre for Theory of Quantum Technologies, University of Gda\'{n}sk, Wita Stwosza 63, 80-308 Gda\'{n}sk, Poland}
\affiliation{Institute of Theoretical Physics and Astrophysics, National Quantum Information Centre, Faculty of Mathematics, Physics and Informatics, University of Gda\'{n}sk, Wita Stwosza 57, 80-308 Gda\'{n}sk, Poland}
\author{Ricard Ravell Rodr{\'i}guez}
\affiliation{International Centre for Theory of Quantum Technologies, University of Gda\'{n}sk, Wita Stwosza 63, 80-308 Gda\'{n}sk, Poland}
\author{Pawe{\l} Horodecki}
\affiliation{International Centre for Theory of Quantum Technologies, University of Gda\'{n}sk, Wita Stwosza 63, 80-308 Gda\'{n}sk, Poland}
\affiliation{Faculty of Applied Physics and Mathematics, National Quantum Information Centre, Gda\'{n}sk University of Technology, Gabriela Narutowicza 11/12, 80-233 Gda\'{n}sk, Poland} 

\begin{abstract}
The problem of achieving security of device-independent (or semi-device-independent) cryptography (for quantum key distribution and randomness generation) against the most general no-signalling adversaries has remained open. It has been recognized that the realisation of extremal non-signalling non-local boxes (or extremal non-signalling non-local assemblages) could provide a route towards devising such highly secure protocols. We first prove a general no-go result that in the Bell non-locality scenario, quantum theory does not allow to realise any extremal non-signalling non-local box, even if scenarios of arbitrary sequential measurements are considered. On the other hand, we secondly prove a positive result showing that a one-sided device-independent scenario where a single party trusts their qubit system is already sufficient for quantum theory to realise a self-testing extremal non-local point within the set of non-signalling assemblages. 
\end{abstract}



\maketitle

\textit{Introduction.-} Correlations in entangled states can not be realised by local hidden variables theories where results of measurements on subsystems are locally predetermined. \cite{EPR1935,Schroedinger1935,H4} This phenomenon evidenced by the violation of Bell inequalities \cite{Bell,RMPBellnonlocality} led to the powerful idea of device-independent (DI) cryptography  \cite{BarrettPRL, PironioNature, PRLAcin, IEEEKessler, Chung} where no assumption on the nature of the quantum systems subject to measurement needs to be made.
In the DI setting, security is ultimately based on the observation of non-local correlations by honest parties and the property of monogamy of quantum non-local correlations \cite{PRLPawlowski,RamanathanPH}. A stronger property than monogamy is that of {\it extremality of the measurement statistics}, i.e., the observation of extremal behavior by honest parties within the set of all measurement behaviors. Such extremal behavior guarantees that their system is completely decoupled from that of any adversary. For any such extremal behavior, one can also find a Bell inequality that is maximally violated by the extremal statistics. 
Moreover, in certain cases, such a violation even permits the \textit{self-testing} \cite{MY04} of the quantum pure state measured i.e., its uniqueness up to irrelevant local operations. 

This analysis can be carried over into general probabilistic theories beyond quantum theory \cite{Barrett2005} that only obey the no-signaling condition of relativity. In this case, there are families of statistics called {\it no-signaling boxes} that obey the no-signaling constraints but may otherwise be super-quantum, and as such may violate Bell inequalities more strongly than quantum boxes, the quintessential example here being the Popescu-Rohrlich (PR) box \cite{Rohlich-Popescu}.
The extremality of a family of statistics in any such no-signaling theory then means that it is uncorrelated from other measurement behaviors (boxes) and as such is very useful in realizing secure DI protocols. 

Later, the weaker scenario of semi-device-independent schemes (semi-DI) has been developed in the setting where some of the
parties may be considered to have full control of the quantum systems in their laboratory (see \cite{Pirandola2019}). Here, instead of just the measurement statistics, one considers \textit{quantum assemblages} and instead of Bell inequalities one considers the so-called steering inequalities  
(see \cite{WJD07}). Similarly, just as the no-signaling boxes, one considers here the \textit {no-signaling assemblages} only constrained by the no-signaling conditions  \cite{SBCSV15}.

The interesting question whether quantum DI cryptography can stay secure against a general no-signaling adversary has been posed \cite{BarrettPRL}. Some partial positive results have been provided in problems of secret key \cite{BarrettPRL} or randomness amplification \cite{ColbeckRennerNaturePhyscis2012, GallegoNatureComm2013, RamanathanNatComm2016}.
These proofs uniformly utilize quantum measurement behaviors that do not represent extremal points in the set of no-signaling behaviors.
It was recognised that if one could realize such extremal post-quantum behaviors by measurements on quantum states, then the security proofs could  
be much more streamlined.
Hence, the natural question was whether there is any scenario in which quantum correlations give rise to extremal no-signaling behaviors.

 
An important, though partial, negative result in this direction was obtained in \cite{RTHHPRL} where it was shown that in the usual Bell non-locality framework, there exists no scenario (number of parties, measurement settings or outcomes) in which quantum correlations  represent an extremal point in the set 
(convex polytope) of no-signaling boxes. An important question was 
left unanswered whether the same is true in more general correlation 
scenarios such as that of sequential Bell non-locality \cite{GWCAN} or in quantum steering scenarios \cite{SBCSV15}.  
 
Here we provide complete answers to both these fundamental questions. First, we extend the no-go result of \cite{RTHHPRL} to
the general scenario of sequential Bell non-locality \cite{GWCAN}: quantum sequential non-local correlations cannot realise extremal non-signaling behaviors, irrespective of the number of measurement settings or outcomes. Second, we also provide a positive answer in the setting of steering inequalities: if one of the parties has a fully trusted qubit system then there exist situations where
quantum assemblages are extremal within general no-signaling assemblages. Crucially, this latter result holds
in the setting of three-party steering where quantum assemblages
have been shown to be a strict subset of the
set of no-signaling assemblages. This surprising result, in view of the unrealisability of super-quantum boxes such as the PR box in non-locality and the interesting consequences thereof \cite{RMPBellnonlocality, RMPBuhrman}, should have interesting consequences both in quantum foundations and in the development of semi-device-independent cryptography secure against no-signaling adversaries.

\textit{Extremality in Sequential Bell non-locality.-} 
We begin with the scenario of sequential Bell non-locality \cite{GWCAN}, where each party performs measurements on their system in a sequential manner, leading to a time-ordered no-signaling structure and the corresponding inequalities consider correlations between outcomes obtained in sequential runs. This scenario is in many ways richer than the usual Bell non-locality scenario, with the appearance of novel phenomena such as 'hidden non-locality' \cite{Popescu}, wherein some quantum states only display local correlations in traditional Bell experiments while exhibiting non-local correlations when correlations are considered also between outcomes of measurements performed in sequence by each party. Here, one party Alice chooses to measure one of $m_A^{(j_A)}$ inputs $i_A^{(j_A)} = 1, \dots, m_A^{(j_A)}$ in the $j_A$-th run of the Bell experiment, and obtains one of $d^{(j_A)}_{A, i_A}$ outputs $o_A^{(j_A)} \in \{1, \dots, d^{(j_A)}_{A, i_A}\}$. Similarly, the other party Bob chooses to measure in the $j_B$-th run, one of $m^{(j_B)}_B$ inputs $i^{(j_B)}_B = 1, \dots, m^{(j_B)}_B $, and obtains one of $d^{(j_B)}_{B, i_B}$ outputs $o^{(j_B)}_B \in \{1, \dots, d^{(j_B)}_{B, i_B}\}$ outputs. Here, $j_A = 1, \dots, N_A$ and $j_B = 1, \dots, N_B$ where $N_A, N_B$ denote the number of measurement runs of Alice and Bob, respectively. Such a sequential Bell scenario is denoted by $\textbf{B}\left(2; (\vec{m}_A, \vec{d}_A); (\vec{m}_B, \vec{d}_B) \right)$, where $\vec{m}_A := (m_A^{(1)}, \dots, m_A^{(N_A)})$, $\vec{d}_A := \left( \left(d^{(1)}_{A,1}, \dots, d^{(1)}_{A, m^{(1)}_A}\right), \dots, \left(d^{(N_A)}_{A,1}, \dots, d^{(N_A)}_{A, m^{(N_A)}_A} \right) \right)$. We will simplify the notation by choosing $m_A^{(j_A)} = m_B^{(j_B)} = m$, $d_A^{(j_A)} = d_B^{(j_B)} = d$ for all $j_A, j_B$, and $N_A = N_B = N$ where this does not affect the generality of the argument. The joint probability of obtaining the outcomes $\textbf{o}_A := (o_A^{(1)}, \dots, o_A^{(N)})$ for Alice, and $\textbf{o}_B := (o_B^{(1)}, \dots, o_B^{(N)})$ for Bob, for given measurement settings $\textbf{i}_A := (i_A^{(1)}, \dots, i_A^{(N)})$ and $\textbf{i}_B := (i_B^{(1)}, \dots, i_B^{(N)})$ respectively, will be denoted by $P_{\textbf{O}_A, \textbf{O}_B | \textbf{I}_A, \textbf{I}_B}(\textbf{o}_A, \textbf{o}_B | \textbf{i}_A, \textbf{i}_B)$. As before, we may view these $n_{seq} := \left(m d \right)^{2N}$ probabilities as forming the components of a vector $P_{\textbf{O}_A, \textbf{O}_B | \textbf{I}_A, \textbf{I}_B} = | P \rangle$ in $\mathds{R}^{n_{seq}}$, and are described as forming a box $P$. 

We consider the set of general Time-Ordered No-Signaling (TONS) boxes in the scenario of sequential non-locality as obeying the time-ordered no-signaling constraints (where there is no-signaling between all rounds of Alice and all rounds of Bob, while signaling is allowed between past rounds of Alice (Bob) to future rounds of Alice (Bob)) in addition to those of normalization and non-negativity, and denote this set as $\textbf{TONS}\left[\textbf{B}\left(2; (\vec{m}_A, \vec{d}_A); (\vec{m}_B, \vec{d}_B) \right)\right]$. The important subset of TONS boxes is the classical Time-Ordered Local Deterministic polytope, denoted by $\textbf{TOLoc}\left[\textbf{B}\left(2; (\vec{m}_A, \vec{d}_A); (\vec{m}_B, \vec{d}_B) \right)\right]$, which is the convex hull of  all boxes where all entries are integral, i.e., in $\{0,1\}$. The boxes obtainable by performing general sequential quantum measurements on a quantum state of arbitrary dimension form the set of sequential quantum correlations denoted by $\textbf{Qseq}\left[\textbf{B}\left(2; (\vec{m}_A, \vec{d}_A); (\vec{m}_B, \vec{d}_B) \right)\right]$. These sets are defined explicitly in the Supplemental Material \cite{Sup}. We ask the question whether quantum correlations can realize the extremal boxes of the general TONS polytope, where an extremal box or a vertex is one that cannot be expressed as a non-trivial convex combination of boxes in the polytope. This fundamental question in quantum foundations gains additional interest in DI quantum cryptography due to the simple but powerful fact that extremal quantum correlations are automatically decoupled from any systems held by any no-signaling adversary \cite{Barrett2005}. By considering an extension to the scenario of sequential non-locality \cite{NPA2, BBS} of the well-known NPA hierarchy \cite{NPA} of semi-definite programming relaxations to the set of quantum correlations, and developing the techniques from \cite{RTHHPRL} to this scenario, we prove (see \cite{Sup}) the following. 

\begin{theorem}
\label{thm:qext-seq}
For  any $(\vec{m}_A, \vec{d}_A), (\vec{m}_B, \vec{d}_B)$ let $P$ be an extremal box of the Time-Ordered No-Signaling polytope $\textbf{TONS}\left[\textbf{B}\left(2; (\vec{m}_A, \vec{d}_A); (\vec{m}_B, \vec{d}_B) \right)\right]$ such that $P \notin \textbf{TOLoc}\left[\textbf{B}\left(2; (\vec{m}_A, \vec{d}_A); (\vec{m}_B, \vec{d}_B) \right)\right]$. Then, $P \notin \text{cl}\left(\textbf{Qseq}\left[\textbf{B}\left(2; (\vec{m}_A, \vec{d}_A); (\vec{m}_B, \vec{d}_B) \right)\right] \right)$. The latter stays true even when the no-signaling constraints are relaxed to allow signaling from the $j$-th run of Alice (Bob) to the $j+k$-th run of Bob (Alice) for all $j = 1, \dots, n$, $k \geq 1$.
\end{theorem} 
Together with the results from \cite{RTHHPRL} the above theorem rules out the quantum realisation of extremal postquantum statistics, at least in the ubiquitous two-party non-locality setting. Nevertheless, subsequently we show below that the situation can be remedied in the steering scenario with the addition of a third party holding a trusted qubit system.

\textit{Extremality of quantum assemblages.-}
Consider a bipartite steering scenario \cite{S36,WJD07} in which two distant subsystems A and B share a quantum state  $\rho^{(AB)}$. We assume that A is uncharacterised (i.e.  dimension of its Hilbert space, reduced quantum state and local measurements which are performed on it are unknown), while the quantum system of B is fully characterised. Let $M^{(A)}_{a|x}$ represent an element of a positive operator valued-measure (POVM) on A, corresponding to the outcome $a \in \mathcal{A}$ of the measurement setting $x \in \mathcal{X}$ with fixed and finite alphabet sizes $|\mathcal{A}|$ and $|\mathcal{X}|$. 
According to measurements performed on A, the subsystem B is then described by the set of subnormalised states
\begin{equation}\label{assemblage}
\sigma^{(B)}_{a|x}=\mathrm{Tr}_{A}\left(M^{(A)}_{a|x}\otimes \mathds{1}\rho^{(AB)}\right).
\end{equation}
The probability of obtaining outcome $a$ while performing measurement $x$ on subsystem A is given by $\mathrm{Tr}_{B}(\sigma^{(B)}_{a|x})$, and subsystem B after this measurement is described by the state $\frac{\sigma^{(B)}_{a|x}}{\mathrm{Tr}_{B}(\sigma^{(B)}_{a|x})}$. The collection of subnormalised states $\Sigma^{(B)}=\left\{\sigma^{(B)}_{a|x}\right\}_{a,x}$ acting on a Hilbert space of dimension $d_B$ is known as a \textit{quantum assemblage}.

One can also consider a general abstract notion of a \textit{no-signaling assemblage} (also acting on a $d_B$ dimensional Hilbert space) defined by the following no-signaling conditions $\forall_{a,x}\ \sigma^{(B)}_{a|x}\geq 0$, $\forall_{x,x'}\ \sum_a\sigma^{(B)}_{a|x}=\sigma^{(B)}=\sum_a\sigma^{(B)}_{a|x'}$ and $\mathrm{Tr}(\sigma^{(B)})=1$.
One can think of such a no-signaling assemblage as the effect of the steering of a quantum state describing subsystem B (with dimension $d_B$) by local measurements performed on an uncharacterised separated subsystem A, when the joint state of both subsystems is no longer described by quantum mechanics, but rather as the state in some no-signaling generalised probabilistic theory. However, it has been proven in  \cite{G89,HJW93}, that any two-party no-signaling assemblage also admits a \textit{quantum realisation}, i.e., there exist a subsystem A, POVM elements $M^{(A)}_{a|x}$ and a joint quantum state $\rho^{(AB)}$, such that all the elements $\sigma^{(B)}_{a|x}$ can be reconstructed as in formula (\ref{assemblage}). Therefore there is no post-quantum steering in this bipartite setting.

The situation dramatically changes if we consider assemblages with three separated subsystems A, B, C, in which a characterised subsystem C associated with a Hilbert space of dimension $d_C$, shares with uncharacterised parties A, B a joint state in some no-signaling generalised probabilistic theory \cite{CS15}. Analogously to the bipartite case, one may perform uncharacterised (local, independent) measurements on A and B (with settings and outcomes labeled by pairs $x,a$ and $y,b$ respectively). System C is then described by a set of subnormalised states $\sigma^{(C)}_{ab|xy}$ satisfying the no-signaling conditions. In this case, the  abstract \textit{no-signaling assemblage} (acting on the $d_C$ dimensional space) $\Sigma^{(C)}=\left\{\sigma^{(C)}_{ab|xy}\right\}_{a,b,x,y}$ is therefore defined by the conditions 
\begin{equation}
\forall_{b,x,x',y}\ \sum_a\sigma^{(C)}_{ab|xy}=\sum_a\sigma^{(C)}_{ab|x'y},
\end{equation}
\begin{equation}
\forall_{a,x,y,y'}\ \sum_b\sigma^{(C)}_{ab|xy}=\sum_b\sigma^{(C)}_{ab|xy'},
\end{equation}
\begin{equation}
\forall_{x,y}\ \mathrm{Tr}\left(\sum_{a,b}\sigma^{(C)}_{ab|xy}\right)=1,\ \forall_{a,b,x,y}\ \sigma^{(C)}_{ab|xy}\geq 0.
\end{equation}

Crucially, as opposed to the bipartite setting, not all no-signaling assemblages $\Sigma^{(C)}$ in the tripartite setting, admit a \textit{quantum realisation} \cite{SBCSV15} as $\sigma^{(C)}_{ab|xy}=\mathrm{Tr}_{AB}\left(M^{(A)}_{a|x}\otimes N^{(B)}_{b|y}\otimes \mathds{1}\rho^{(ABC)} \right)$, with POVM elements $M^{(A)}_{a|x}, N^{(B)}_{b|y}$ and tripartite state $\rho^{(ABC)}$ of the quantum system ABC. 

Indeed, one may consider a no-signaling assemblage defined as $\sigma^{(C)}_{ab|xy}=p^{(AB)}(ab|xy)\rho^{(C)}$ with $p^{(AB)}(ab|xy)$ denoting the so-called Popescu-Rohrlich (PR) box distributions \cite{SBCSV15}. This assemblage is post-quantum and this is a direct consequence of the post-quantum non-locality of the PR box \cite{Sup}. Interestingly, it has been found that there are also no-signaling assemblages $\Sigma^{(C)}=\left\{\sigma^{(C)}_{ab|xy}\right\}_{a,b,x,y}$, for which any POVM elements $R^{(C)}_{c|z}$ provide no-signaling boxes $p^{(ABC)}(abc|xyz)=\mathrm{Tr}\left(R^{(C)}_{c|z}\sigma^{(C)}_{ab|xy}\right)$ with quantum realisation, and yet the whole assemblage does not admit quantum realisation \cite{SBCSV15}. These show that postquantum non-locality and post-quantum steering are genuinely different phenomena in the tripartite setting and beyond. It is noteworthly that the (i) set of no-signaling assamblages and (ii) the subset of assemblages that admit quantum realisation are both convex \cite{Sup}.

Inside the discussed set of quantum assemblages one can single out the convex subset of LHS (local hidden state) assemblages which represent steering with a classically correlated system \cite{SAPHS18}. A no-signaling assemblage admits \textit{LHS model} if it can be represented by $\sigma^{(C)}_{ab|xy}=\sum_i q_i p^{(A)}_i(a|x)p^{(B)}_i(b|y)\sigma^{(C)}_i$ where $q_i\geq 0, \sum_i q_i=1$, and $\sigma^{(C)}_i$ are some states of characterised subsystem C and $p^{(A)}_i(a|x),p^{(B)}_i(b|y)$ denote conditional probability distributions for uncharacterised subsystems A and B respectively. Equivalently, for LHS $\sigma^{(C)}_{ab|xy}=\sum_i q_ip^{(AB)}_i(ab|xy)\sigma^{(C)}_i$ where $L_i=\left\{p^{(AB)}_i(ab|xy)\right\}_{a,b,x,y}$ is a deterministic box of conditional probabilities \cite{Sup}. 

 As in a tripartite case one can discuss different type of separability (entanglement), we introduce another convex set of \textit{biseparable assemblages} (BIS) as a collection of all assemblages with quantum realisation $\sigma^{(C)}_{ab|xy}=\mathrm{Tr}_{AB}\left(M^{(A)}_{a|x}\otimes N^{(B)}_{b|y}\otimes \mathds{1}\rho^{(ABC)} \right)$ where $\rho^{(ABC)}$ is biseparable \cite{CS15} (see further discussion in Supp. Material \cite{Sup}). It is easy to see that biseparable assemblages form an intermediate set between LHS and quantum assemblages.\color{black}.

One can show that a no-signaling assemblage can be excluded from the set of LHS assemblages by the violation of a \textit{steering inequality}, i.e. for any no-signaling assemblage $\Sigma^{(C)}$ that does not admit an LHS model, there exists a linear real valued functional $F$ on no-signaling assemblages such that $F(\Sigma^{(C)})>c$ and $F(\Sigma^{(C)}_{LHS})\leq c$ for all LHS assemblages $\Sigma^{(C)}_{LHS}$. Similarly certain subclass of such inequalities may be used for certification that a given assemblage is not biseparable. In particular, in case of quantum assemblages, steering inequities may indicate that the initial state is not fully separable or biseprarable. \color{black}

One can easily generalise the notion of no-signaling assemblages to the scenario with $n>2$ uncharacterised parties \cite{SAPHS18}. For simplicity, we will restrict our attention to the case when $n=2$ and $a,b,x,y\in \left\{0,1\right\}$. Note that a no-signaling assemblage can be then seen as a box of positive operators (i.e. subnormalised states) where $(a|x)$ label rows and $(b|y)$ label columns, i.e. 
\begin{equation}\label{eq:NS-assemblage}
\Sigma^{(C)}=\begin{pmatrix}
\begin{array}{cc|cc}
 \sigma^{(C)}_{00|00} &  \sigma^{(C)}_{01|00} & \sigma^{(C)}_{00|01} &  \sigma^{(C)}_{01|01} \\  
 \sigma^{(C)}_{10|00} & \sigma^{(C)}_{11|00}& \sigma^{(C)}_{10|01}  & \sigma^{(C)}_{11|01} \\ \hline
 \sigma^{(C)}_{00|10} & \sigma^{(C)}_{01|10} & \sigma^{(C)}_{00|11}&  \sigma^{(C)}_{01|11}  \\
   \sigma^{(C)}_{10|10} & \sigma^{(C)}_{11|10} & \sigma^{(C)}_{10|11} & \sigma^{(C)}_{11|11}
\end{array}
\end{pmatrix}.
\end{equation}In particular, LHS assemblages are convex combinations of extremal boxes  (of operators) which have only four nonzero positions occupied by the same pure state forming a rectangle with exactly one element for each pair $(x,y)$ (see example (B20) in \cite{Sup}).

\color{black}

 In analogy to the fundamental question in nonlocality it is interesting to ask whether a quantum assemblage can realize an extremal non-classical point in the set of no-signaling assemblages. In the case of bipartite steering all no-signaling assemblages admit quantum realisation, therefore such a question is uninteresting. The first relevant scenario is a tripartite setup with at least two measurement settings on uncharacterised parties \cite{Sup}. We show below that, remarkably, in contrast to non-locality the question admits the affirmative answer in this setting.

Below we shall introduce the concepts of {\it similarity} and  {\it inflexibility} of no-signaling assemblages with pure rank one elements $\sigma_{ab|xy}^{(C)}$, which is crucial for further analysis.

\begin{definition}\label{similarity}
Consider a general no-signaling assemblage $\Sigma^{(C)}$ as in Eq.(\ref{eq:NS-assemblage}) with all positions occupied by at most rank one operators and denote it by $\Sigma^{(C)}=\left\{p_i|\psi^{(C)}_i\rangle \langle \psi^{(C)}_i|\right\}_i$, where $i=(ab|xy)$ and $p_i=\mathrm{Tr}(\sigma^{(C)}_i)$.  Any other assemblage $\tilde{\Sigma}^{(C)}=\left\{q_i|\psi^{(C)}_i\rangle \langle \psi^{(C)}_i|\right\}_i$ with the same states at the same positions and the additional property that $p_i=0$ implies $q_i=0$ is called {\it similar} to $\Sigma^{(C)}$.
\end{definition}
Note that the above relation is not symmetric, i.e. it may happen that $\tilde{\Sigma}^{(C)}$ is similar to $\Sigma^{(C)}$, 
but $\Sigma^{(C)}$ is not similar to $\tilde{\Sigma}^{(C)}$. The second concept is defined here

\begin{definition}\label{inflexibility}
An assemblage $\Sigma^{(C)}$ is called \textit{inflexible} if for any $\tilde{\Sigma}^{(C)}$ similar to $\Sigma^{(C)}$ we 
get $\Sigma^{(C)}=\tilde{\Sigma}^{(C)}$.
\end{definition}

Note that in particular {\it inflexibility implies extremality} in the set of all no-signaling assemblages (see Section B2 in \textit{Sup. Material} \cite{Sup}).\color{black}

Any extremal quantum assemblage can be obtained by measurements performed on a pure state, therefore we may restrict  only to such states \cite{Sup}. Recall that a pure state $|\psi^{(ABC)}\rangle\in \mathbb{C}^{d_A}\otimes \mathbb{C}^{d_B}\otimes \mathbb{C}^{d_C}$ is genuine three-party entangled if it is entangled with respect to any bipartite splitting of the tripartite system. 

\begin{proposition}\label{proposition_genuine}
For any pure genuine three-party entangled state $|\psi^{(ABC)}\rangle\in \mathbb{C}^{2}\otimes \mathbb{C}^{2}\otimes \mathbb{C}^{d}$ there exists a pair of PVMs with two outcomes on subsystems A and B respectively such that a no-signaling assemblage $\Sigma^{(C)}$ obtained by these measurements is inflexible. In particular $\Sigma^{(C)}$ is extremal, not LHS and not biseparable. Moreover, $\Sigma^{(C)}$ is the unique no-signaling assemble that maximally violates some steering inequality $F_{\Sigma^{(C)}}$.\color{black}
\end{proposition}
\begin{sproof}
Since $|\psi^{(ABC)}\rangle$ is genuine three-party entangled there exists a PVM with elements $P^{(A)}_{a|0} =|\phi_{a|0}\rangle \langle \phi_{a|0}|$ on the subsystem A such that 
$\langle \phi_{a|0}|\psi^{(ABC)}\rangle$ are entangled and linearly independent (see Lemma 4 and 5 in \cite{Sup}). Therefore, one can choose a pair of PVMs with respective elements $Q^{(B)}_{b|0},Q^{(B)}_{b|1}$ on the subsystem B, such that each of the first two rows of $\Sigma^{(C)}$  consists of elements  proportional to normalised pure states which are all different. Moreover, there is an index $(b|y)$ for which $\sigma^{(C)}_{0b|0y}$ and $\sigma^{(C)}_{1b|0y}$ are not proportional to the same pure state (see Lemma 6 in \cite{Sup}). Choosing the second PVM on subsystem A such that its elements $P^{(A)}_{a|1}$ do not commute with the first, we obtain a column $(b|y)$ with the same property as the first and the second row i.e. 
having all the normalised elements pure and different. A detailed analysis of assemblages with such properties proves that $\Sigma^{(C)}$ is inflexible and hence extremal (see Supp. Material \cite{Sup}). Define now
\begin{equation}
 \rho_{ab|xy}=
\begin{cases}
0\ \ \ \mathrm{for}\ \sigma^{(C)}_{ab|xy}=0,\\
\frac{\sigma^{(C)}_{ab|xy}}{\mathrm{Tr}(\sigma^{(C)}_{ab|xy})} \ \ \  \mathrm{for}\ \sigma^{(C)}_{ab|xy}\neq 0.
\end{cases}
\end{equation}For any no-signaling assemblage $\tilde{\Sigma}^{(C)}$ consider
\begin{equation}\label{expression}
F_{\Sigma^{(C)}}(\tilde{\Sigma}^{(C)})=\sum_{a,b,x,y = 0,1}\mathrm{Tr}(\rho_{ab|xy}\tilde{\sigma}^{(C)}_{ab|xy}).
\end{equation}Observe that $F_{\Sigma^{(C)}}(\tilde{\Sigma}^{(C)})\leq 4$ and equality holds if and only if the no-signaling assemblage $\tilde{\Sigma}^{(C)}$ is similar to $\Sigma^{(C)}$, so by inflexibility the  maximal value of $F_{\Sigma^{(C)}}$ is uniquely obtained for $\Sigma^{(C)}$. Since any LHS assemblage is a convex combination of assemblages consisting of  {\it the same} pure state occupying four positions forming a rectangle, and the first and second row of $\Sigma^{(C)}$ consist of {\it different} rank one operators, $\Sigma^{(C)}$ is not an LHS assemblage and $\mathcal{C}_{\mathbf{LHS}}=\sup_{LHS}F_{\Sigma^{(C)}}(\Sigma^{(C)}_{LHS})<4$. Analyzing structure of the set of biseparale assemblages one can additionally show that $\mathcal{C}_{\mathbf{BIS}}=\sup_{BIS}F_{\Sigma^{(C)}}(\Sigma^{(C)}_{BIS})<4$ as form of $\Sigma^{(C)}$ dose not agree with possible form of extremal no-signaling assemblage that is biseparable (see \cite{Sup}).\color{black}
\end{sproof}

To find $\mathcal{C}_{\mathbf{LHS}}$\color{black}, the value $\sum_{a,b,x,y\in I(L)}\mathrm{Tr}(\rho_{ab|xy}|\phi\rangle \langle \phi|)$ is maximized over pure states $|\phi\rangle$ and deterministic boxes $L$, where $I(L)$ denotes the set of $(ab|xy)$ (forming a rectangle) for which $p(ab|xy)=1$ in $L$ - by convexity, optimization need be performed only over the extremal points \cite{Sup}. Optimization over biseparabe assemblages boils down to optimization over three types of assemblages - compare with \cite{Sup}.\color{black}


\begin{example}
\label{eq:Example-GHZ}
Consider a GHZ three qubit state $|\psi^{(ABC)}\rangle=\frac{1}{\sqrt{2}}\left(|000\rangle +|111\rangle\right)$. Let $\Sigma^{(C)}_{GHZ}$ be given by $\sigma^{(C)}_{ab|xy}=\mathrm{Tr}_{AB}(P^{(A)}_{a|x}\otimes Q^{(B)}_{b|y}\otimes \mathds{1}|\psi^{(ABC)}\rangle \langle \psi^{(ABC)}|)$ with $P^{(A)}_{0|0}=Q^{(B)}_{0|0}=|+\rangle \langle +|$ and $P^{(A)}_{0|1}=Q^{(B)}_{0|1}=|0\rangle \langle 0|$, i.e.
\begin{equation}\nonumber
\Sigma^{(C)}_{GHZ}=\frac{1}{4}\begin{pmatrix}
\begin{array}{cc|cc}
 |+\rangle \langle +| &  |-\rangle \langle -| & |0\rangle \langle 0| &  |1\rangle \langle 1| \\  
 |-\rangle \langle -| & |+\rangle \langle +|& |0\rangle \langle 0|  &|1\rangle \langle 1|\\ \hline
 |0\rangle \langle 0| & |0\rangle \langle 0| &2|0\rangle \langle 0|&  0  \\
   |1\rangle \langle 1| &  |1\rangle \langle 1|& 0 & 2|1\rangle \langle 1| 
\end{array}
\end{pmatrix}.
\end{equation}Assemblage $\Sigma^{(C)}_{GHZ}$ is inflexible by Proposition (\ref{proposition_genuine}) and 
$
\mathcal{C}_{\mathbf{LHS}}=\sup_{|\phi\rangle}\mathrm{Tr}\left[(3|0\rangle \langle 0|+|+\rangle \langle +|)|\phi\rangle \langle \phi|\right]=\frac{4+\sqrt{10}}{2}
$ while

$
\mathcal{C}_{\mathbf{BIS}}=\frac{5+\sqrt{5}}{2}
$(see detailed calculation in Supp. Material \cite{Sup}).
\end{example}\color{black}

To investigate the result further, let us fix $|\psi^{(ABC)}\rangle$ together with its related $\Sigma^{(C)}$ obtained using PVMs $P^{(A)}_{a|x}, Q^{(B)}_{b|y}$ as in Proposition \ref{proposition_genuine}. Consider an arbitrary pure state $|\tilde{\psi}^{(ABC)}\rangle\in \mathbb{C}^{2}\otimes \mathbb{C}^{2}\otimes \mathbb{C}^{d}$ and corresponding assemblage $\tilde{\Sigma}^{(C)}$ given by two pairs of PVMs $\tilde{P}^{(A)}, \tilde{Q}^{(B)}$ with two outcomes as $\tilde{\sigma}^{(C)}_{ab|xy}=\mathrm{Tr}_{AB}(\tilde{P}^{(A)}_{a|x}\otimes \tilde{Q}^{(B)}_{b|y}\otimes \mathds{1}|\tilde{\psi}^{(ABC)}\rangle \langle \tilde{\psi}^{(ABC)}|)$. One can see that $F_\Sigma(\tilde{\Sigma})=4$ iff $|\tilde{\psi}^{(ABC)}\rangle=U_A\otimes U_B\otimes \mathds{1}|\psi^{(ABC)}\rangle$, $\tilde{P}^{(A)}_{a|x}=U_AP^{(A)}_{a|x}U_A^\dagger$ and $\tilde{Q}^{(B)}_{b|y}=U_BQ^{(B)}_{b|y}U_B^\dagger$ for some local unitaries $U_A,U_B$. This leads to the following \textit{self-testing} result (see \cite{Sup} for the proof).

\begin{proposition}
For any pure state $|\tilde{\psi}^{(A'B'C)}\rangle\in \mathbb{C}^{d_{A'}}\otimes \mathbb{C}^{d_{B'}}\otimes \mathbb{C}^{d_{C}}$ and assemblage $\tilde{\Sigma}^{(C)}$ with elements $\tilde{\sigma}^{(C)}_{ab|xy}=\mathrm{Tr}_{A'B'}\left(\tilde{P}^{(A')}_{a|x}\otimes \tilde{Q}^{(B')}_{b|y}\otimes \mathds{1}|\tilde{\psi}^{(A'B'C)}\rangle \langle \tilde{\psi}^{(A'B'C)}| \right)$, $F_{\Sigma^{(C)}}(\tilde{\Sigma}^{(C)})=4$ iff $V_{A'} \otimes V_{B'} \otimes \mathds{1}|\tilde{\psi}^{(A'B'C)}\rangle=|\phi^{junk}_{A''B''}\rangle|\psi^{(ABC)}\rangle$, and $(V_{A'} \otimes V_{B'} \otimes \mathds{1})(\tilde{P}^{(A')}_{a|x}\otimes \tilde{Q}^{(B')}_{b|y}\otimes \mathds{1})|\tilde{\psi}^{(A'B'C)}\rangle=|\phi_{A''B''}\rangle(P^{(A)}_{a|x}\otimes Q^{(A)}_{b|y}\otimes \mathds{1})|\psi^{(ABC)}\rangle$ where $V_{A'},V_{B'}$ are some local isometries and $|\psi^{(ABC)}\rangle$ together with $P^{(A)}_{a|x},  Q^{(A)}_{b|y}$ are as in Proposition \ref{proposition_genuine}, while $|\phi^{junk}_{A'' B''} \rangle$ is some irrelevant junk state. \color{black}
\end{proposition}

{\it Discussion.-} In this paper we have proved that in the most general scenario of sequential measurements it is impossible to quantumly realise non-local extremal points of the Time-Ordered No-Signaling polytope. This answers the open question posed in \cite{RTHHPRL}. On the contrary, if one of the parties has access to at least one fully trusted qubit system, we have shown that one can obtain quantum assemblages which are extremal within the set of no-signaling assemblages. While this opens the path towards security proofs for semi-device-independent cryptographic protocols against general adversaries, numerous interesting open questions arise for future research. In particular, in the setting of sequential Bell non-locality, the immediate question is to extend the result to the multi-partite setting, as well as to the scenario of single-party contextuality. Quantitative bounds on the distance of quantum boxes from extremal time-ordered no-signaling ones should be obtained utilizing the methods involved in the proof of Theorem \ref{thm:qext-seq}, i.e., by lower bounding the minimum eigenvalue of the matrix $\tilde{A}^{\top}_P \tilde{A}_P$ associated with the extremal box $P$.  Is it possible to generalise the main result of the Proposition \ref{proposition_genuine} by showing that 
for any genuinely entangled 3-party $d_{A} \otimes d_{B} \otimes d_{C}$ state, 
there are some $k_{A}$ PVMs (or POVMs) with $s_{A}$ outcomes on system A and $k_{B}$ PVMs (or POVMs) with $s_{B}$ outcomes
on system B such that the corresponding assemblage is again extremal in the set of all no-signaling assemblages? If yes, what are the minimal number of settings and outcomes? Clear extensions to many-party scenarios should naturally be explored. 
Finally, as clearly not all extremal no-signaling assemblages admit quantum realisation (e.g. $\sigma^{(C)}_{ab|xy}:=p^{(AB)}(ab|xy)|\phi^{(C)}\rangle \langle \phi^{(C)}|$ with $p^{(AB)}(ab|xy)$ coming from a PR-box), it would be natural to ask for a characterisation of such extremal points and information-theoretic consequences thereof.


{\it Acknowledgments.-}
R. R. acknowledges support from the Start-up Fund "Device-Independent Quantum Communication Networks" from The University of Hong Kong, the Seed Fund "Security of Relativistic Quantum Cryptography" (Grant No. $201909185030$) and the Early Career Scheme (ECS) grant "Device-Independent Random Number Generation and Quantum Key Distribution with Weak Random Seeds" (Grant No. $27210620$). This work was supported by the National Natural Science Foundation of China through grant 11675136, the Hong Kong Research Grant Council through grant 17300918, and the John Templeton Foundation through grants 60609, Quantum Causal Structures, and 61466, The Quantum Information Structure of Spacetime (qiss.fr). The opinions expressed in this publication are those of the authors and do not necessarily reflect the views of the John Templeton Foundation. M.B., R.R.R. and P.H. acknowledge support by the Foundation for Polish Science (IRAP project, ICTQT, contract no. 2018/MAB/5, co-financed by EU within the Smart Growth Operational Programme).


{\it Competing Interests.-} The authors declare that there are no competing interests.

{\it  Data Availability.-} Data available within the article and supplementary materials.

{\it Author Contributions.-} The authors contributed equally to this work.


\section{Supplemental Material: Extremal non-signaling non-local boxes cannot be realized in quantum theory even in sequential Bell scenarios }


Here, we give the formal proofs of the Theorem $1$ from the main text that quantum theory does not allow for the realisation of extremal non-local non-signalling boxes even in sequential Bell scenarios. 


\subsection{Preliminaries}

While the scenario of sequential non-locality is the subject of this section, for didactical purposes we shall start with the usual two-party Bell non-locality scenario in this preliminary section.

Consider a two-party Bell experiment. Suppose one party Alice chooses to measure one of $m_A$ inputs $i_A = 1, \dots, m_A$, and obtains one of $d_{A, i_A}$ outputs $o_A \in \{1, \dots, d_{A, i_A}\}$. Similarly, the other party Bob chooses to measure one of $m_B$ inputs $i_B = 1, \dots, m_B $, and obtains one of $d_{B, i_B}$ outputs $o_B \in \{1, \dots, d_{B, i_B}\}$ outputs. Such a Bell scenario is denoted by $\textbf{B}(2; m_A, \vec{d}_A; m_B, \vec{d}_B)$, where $\vec{d}_A = (d_{A,1}, \dots, d_{A, m_A})$ and $\vec{d}_B = (d_{B,1}, \dots, d_{B, m_B})$. This notation can also be shortened to $\textbf{B}(\vec{d}_A, \vec{d}_B)$ for simplicity. The joint probability of obtaining the outcomes $(o_A, o_B)$ given the measurement settings $(i_A, i_B)$ is denoted as $P_{O_A, O_B | I_A, I_B}(o_A, o_B | i_A, i_B)$. We may view these $n = \left(\sum_{i_A = 1}^{m_A} d_{A, i_A} \right) \left(\sum_{i_B = 1}^{m_B} d_{B, i_B} \right)$ probabilities as forming the components of a vector $P_{O_A, O_B | I_A, I_B} = | P \rangle$ in $\mathds{R}^n$, where the inputs and outputs are implicit, and the probabilities are also described as forming a box $P$.

The box $P$ is a valid normalized no-signaling box, satisfying the no-signaling constraints of relativity and the normalization of probabilities, if it obeys the constraints of
\begin{enumerate}
\item Non-negativity: $P_{O_A, O_B | I_A, I_B}(o_A, o_B | i_A, i_B) \geq 0$ for all $o_A, o_B, i_A, i_B$, 

\item Normalization: $\sum_{o_A = 1}^{d_{A, i_A}} \sum_{o_B = 1}^{d_{B, i_B}} P_{O_A, O_B | I_A, I_B}(o_A, o_B | i_A, i_B)  = 1$ for all $i_A, i_B$,

\item No-Signaling: 
\begin{widetext}
\begin{eqnarray}
\sum_{o_A = 1}^{d_{A, i_A}} P_{O_A, O_B | I_A, I_B}(o_A, o_B | i_A, i_B) &=& \sum_{o'_A = 1}^{d_{A, i'_A}} P_{O_A, O_B | I_A, I_B}(o'_A, o_B | i'_A, i_B), \quad \text{for all} \; \; i_A, i'_A, o_B, i_B, \nonumber \\
\sum_{o_B = 1}^{d_{B, i_B}} P_{O_A, O_B | I_A, I_B}(o_A, o_B | i_A, i_B) &=& \sum_{o'_B = 1}^{d_{B, i'_B}} P_{O_A, O_B | I_A, I_B}(o_A, o'_B | i_A, i'_B), \quad \text{for all} \; \; i_B, i'_B, o_A, i_A.
\end{eqnarray}
\end{widetext}
\end{enumerate}

The convex hull of all boxes $P$ that satisfy the above constraints forms the No-Signaling Polytope of the Bell scenario $\textbf{NS}\left[\textbf{B}(2; m_A, \vec{d}_A; m_B, \vec{d}_B)\right]$. 

The boxes within the No-Signaling polytope that additionally satisfy the integrality constraint given by
\begin{enumerate}
\setcounter{enumi}{3}
\item Integrality $P_{O_A, O_B | I_A, I_B}(o_A, o_B | i_A, i_B) \in  \{0,1\}$ for all $o_A, o_B, i_A, i_B$,
\end{enumerate}
are said to be Local Deterministic Boxes (LDBs). The convex hull of these LDBs forms the classical or Bell polytope denoted by $\textbf{C}\left[\textbf{B}(2; m_A, \vec{d}_A; m_B, \vec{d}_B)\right]$. This is the set of all correlations obtainable from local hidden variable theories. 

The set of Quantum Correlations denoted by $\textbf{Q}\left[\textbf{B}(2; m_A, \vec{d}_A; m_B, \vec{d}_B)\right]$ also lies within the No-Signaling polytope. This set consists of boxes $P$ where each component $P_{O_A, O_B | I_A, I_B}(o_A, o_B | i_A, i_B)$ is obtained as:
\begin{equation}
\label{eq:qcorr-def}
P_{O_A, O_B | I_A, I_B}(o_A, o_B | i_A, i_B) = \text{Tr}\left[ \rho \left(E^{A}_{i_A, o_A} \otimes E^{B}_{i_B, o_B} \right) \right]
\end{equation}
for some quantum state $\rho \in \mathcal{H}_d$ of some arbitrary dimension $d$, and sets of projection operators $\{E^{A}_{i_A, o_A}\}$ for Alice and $\{E^{B}_{i_B, o_B}\}$ for Bob. Notably, the measurement operators satisfy the requirements of (i) Hermiticity: $\left(E^{A}_{i_A, o_A}\right)^{\dagger} = E^{A}_{i_A, o_A}$ for all $i_A, o_A$, and $\left(E^{B}_{i_B, o_B}\right)^{\dagger} = E^{B}_{i_B, o_B}$, for all $i_B, o_B$, (ii) Orthogonality: $E^{A}_{i_A, o_A} E^{A}_{i_A, o'_A} = \delta_{o_A, o'_A} E^{A}_{i_A, o_A}$ for all $i_A$, and $E^{B}_{i_B, o_B} E^{B}_{i_B, o'_B} = \delta_{o_B, o'_B} E^{B}_{i_B, o_B}$ for all $i_B$, and (iii) Completeness: $\sum_{o_A} E^{A}_{i_A, o_A} = \mathds{1}$ for all $i_A$ and $\sum_{o_B} E^{B}_{i_B, o_B} = \mathds{1}$ for all $i_B$. The set $\textbf{Q}\left[\textbf{B}(2; m_A, \vec{d}_A; m_B, \vec{d}_B)\right]$ is convex but not a polytope. We have the inclusions $\textbf{C}\left[\textbf{B}(2; m_A, \vec{d}_A; m_B, \vec{d}_B)\right] \subseteq \textbf{Q}\left[\textbf{B}(2; m_A, \vec{d}_A; m_B, \vec{d}_B)\right] \subseteq \textbf{NS}\left[\textbf{B}(2; m_A, \vec{d}_A; m_B, \vec{d}_B)\right]$.

In \cite{NPA}, a hierarchy of semi-definite programs was formulated for optimization with non-commuting variables, and this Navascues-Pironio-Acin (NPA) hierarchy is ubiquitously employed to efficiently determine upper bounds to the quantum violation for general Bell inequalities. The hierarchy was also shown to converge to a set $\textbf{Q}^{pr}\left[\textbf{B}(2; m_A, \vec{d}_A; m_B, \vec{d}_B)\right]$, which is the set consisting of boxes $P$ where each component $P_{O_A, O_B | I_A, I_B}(o_A, o_B | i_A, i_B)$ is obtained as:
\begin{equation}
P_{O_A, O_B | I_A, I_B}(o_A, o_B | i_A, i_B) = \text{Tr}\left[ \rho \left(E^{A}_{i_A, o_A}  E^{B}_{i_B, o_B} \right) \right],
\end{equation}
with $\left[E^{A}_{i_A, o_A}, E^{B}_{i_B, o_B} \right] = 0$ for all $i_A, o_A, i_B, o_B$. The above differs from Eq.(\ref{eq:qcorr-def}) in that the strict requirement of tensor product structure is replaced with the requirement of only commutation between different parties' measurements. It is clear that $\textbf{Q}\left[\textbf{B}(2; m_A, \vec{d}_A; m_B, \vec{d}_B)\right] \subseteq \textbf{Q}^{pr}\left[\textbf{B}(2; m_A, \vec{d}_A; m_B, \vec{d}_B)\right]$. 

In the NPA hierarchy, one considers sets consisting of sequences of product projection operators $S_1 = \{ \mathds{1} \} \cup \{E^{A}_{i_A, o_A} \} \cup \{E^{B}_{i_B, o_B}\}$, $S_2 = S_1 \cup \{E^{A}_{i_A, o_A} E^{B}_{i_B, o_B}\}$, etc. The convex sets corresponding to different levels of this hierarchy $\textbf{Q}_{l}\left[\textbf{B}(2; m_A, \vec{d}_A; m_B, \vec{d}_B)\right]$ are constructed by testing for the existence of a certificate $\Gamma^{l}$ associated to the set of operators $S_l$ by means of a semi-definite program. This certificate $\Gamma^{l}$ corresponding to level $l$ of the NPA hierarchy is a $\left|S_l \right| \times |S_l|$ matrix whose rows and columns are indexed by the operators in the set $S_l$. The certificate $\Gamma^{l}$ is required to be a complex Hermitian positive semi-definite matrix satisfying the following constraints on its entries: (i) $\Gamma^{l}_{\mathds{1}, \mathds{1}} = 1$, and (ii) $\Gamma^{l}_{Q, R} = \Gamma^{l}_{S, T}$ if $Q^{\dagger} R = S^{\dagger} T$. The latter condition in particular imposes that $\Gamma^{l}_{\mathds{1}, E^{A}_{i_A, o_A} E^{B}_{i_B, o_B}} = \Gamma^{l}_{E^{A}_{i_A, o_A}, E^{B}_{i_B, o_B}} = \Gamma^{l}_{E^{A}_{i_A, o_A} E^{B}_{i_B, o_B}, E^{A}_{i_A, o_A} E^{B}_{i_B, o_B}} = P_{O_A, O_B | I_A, I_B}(o_A, o_B | i_A, i_B)$. \newline

\subsection{No quantum realization of sequential non-local no-signaling boxes}
\label{sec:Qext-seqnl}
In general, one may consider the scenario of sequential Bell non-locality, where each party performs measurements on their system in a sequential manner in multiple runs of the Bell experiment, leading to a time-ordered no-signaling structure. This scenario is in many ways richer than the single-run Bell non-locality scenario, with the appearance of novel phenomena such as 'hidden non-locality' \cite{Popescu}, wherein some quantum states only display local correlations in single-run Bell experiments while exhibiting non-local correlations when two measurements are performed in sequence by each party. Now, one party Alice chooses to measure one of $m_A^{(j_A)}$ inputs $i_A^{(j_A)} = 1, \dots, m_A^{(j_A)}$ in the $j_A$-th run of the Bell experiment, and obtains one of $d^{(j_A)}_{A, i_A}$ outputs $o_A^{(j_A)} \in \{1, \dots, d^{(j_A)}_{A, i_A}\}$. Similarly, the other party Bob chooses to measure in the $j_B$-th run, one of $m^{(j_B)}_B$ inputs $i^{(j_B)}_B = 1, \dots, m^{(j_B)}_B $, and obtains one of $d^{(j_B)}_{B, i_B}$ outputs $o^{(j_B)}_B \in \{1, \dots, d^{(j_B)}_{B, i_B}\}$ outputs. Here, $j_A = 1, \dots, N_A$ and $j_B = 1, \dots, N_B$ where $N_A, N_B$ denote the number of measurement runs of Alice and Bob, respectively. Such a sequential Bell scenario is denoted by $\textbf{B}\left(2; (\vec{m}_A, \vec{d}_A); (\vec{m}_B, \vec{d}_B) \right)$, where $\vec{m}_A := (m_A^{(1)}, \dots, m_A^{(N_A)})$, $\vec{d}_A := \left( \left(d^{(1)}_{A,1}, \dots, d^{(1)}_{A, m^{(1)}_A}\right), \dots, \left(d^{(N_A)}_{A,1}, \dots, d^{(N_A)}_{A, m^{(N_A)}_A} \right) \right)$. We will simplify the notation by choosing $m_A^{(j_A)} = m_B^{(j_B)} = m$, $d_A^{(j_A)} = d_B^{(j_B)} = d$ for all $j_A, j_B$, and $N_A = N_B = N$ where this does not affect the generality of the argument. The joint probability of obtaining the outcomes $\textbf{o}_A := (o_A^{(1)}, \dots, o_A^{(N)})$ for Alice, and $\textbf{o}_B := (o_B^{(1)}, \dots, o_B^{(N)})$ for Bob, for given measurement settings $\textbf{i}_A := (i_A^{(1)}, \dots, i_A^{(N)})$ and $\textbf{i}_B := (i_B^{(1)}, \dots, i_B^{(N)})$ respectively, will be denoted by $P_{\textbf{O}_A, \textbf{O}_B | \textbf{I}_A, \textbf{I}_B}(\textbf{o}_A, \textbf{o}_B | \textbf{i}_A, \textbf{i}_B)$. As before, we may view these $n_{seq} := \left(m d \right)^{2N}$ probabilities as forming the components of a vector $P_{\textbf{O}_A, \textbf{O}_B | \textbf{I}_A, \textbf{I}_B} = | P \rangle$ in $\mathds{R}^{n_{seq}}$, and are described as forming a box $P$. 

We consider the general time-ordered no-signaling boxes in the scenario of sequential non-locality as obeying the time-ordered no-signaling constraints in addition to those of normalization and non-negativity. A box $P$ is a valid time-ordered no-signaling box, if it obeys the constraints:  
\begin{enumerate}
\item Non-negativity: $P_{\textbf{O}_A, \textbf{O}_B | \textbf{I}_A, \textbf{I}_B}(\textbf{o}_A, \textbf{o}_B | \textbf{i}_A, \textbf{i}_B) \geq 0$ for all $\textbf{o}_A, \textbf{o}_B, \textbf{i}_A, \textbf{i}_B$, 

\item Normalization: $\sum_{\textbf{o}_A, \textbf{o}_B}  P_{\textbf{O}_A, \textbf{O}_B | \textbf{I}_A, \textbf{I}_B}(\textbf{o}_A, \textbf{o}_B | \textbf{i}_A, \textbf{i}_B)  = 1$ for all $\textbf{i}_A, \textbf{i}_B$,

\item Time-Ordered No-Signaling: 
\begin{widetext}
\begin{eqnarray}
\sum_{o_A^{(k)}, \dots, o_A^{(N)}} P_{\textbf{O}_A, \textbf{O}_B | \textbf{I}_A, \textbf{I}_B}(\textbf{o}_A, \textbf{o}_B | \textbf{i}_A, \textbf{i}_B), \quad \text{is independent of} \; i_A^{(k)}, \dots, i_A^{(N)} \; \text{for all} \; \; 1 \leq k \leq N, \textbf{i}_A, \textbf{o}_A, \textbf{i}_B, \textbf{o}_B, \nonumber \\
\sum_{o_B^{(k)}, \dots, o_B^{(N)}} P_{\textbf{O}_A, \textbf{O}_B | \textbf{I}_A, \textbf{I}_B}(\textbf{o}_A, \textbf{o}_B | \textbf{i}_A, \textbf{i}_B), \quad \text{is independent of} \; i_B^{(k)}, \dots, i_B^{(N)} \; \text{for all} \; \; 1 \leq k \leq N, \textbf{i}_A, \textbf{o}_A, \textbf{i}_B, \textbf{o}_B.
\end{eqnarray}
\end{widetext}
\end{enumerate}
The above time-ordered no-signaling constraints capture the restriction that measurement outcomes for each party do not input on the future inputs chosen by that party, as well as any of the inputs chosen by the other party. The convex hull of all boxes $P$ that satisfy the above constraints forms the Time-Ordered No-Signaling Polytope of the Bell scenario and is denoted by $\textbf{TONS}\left[\textbf{B}\left(2; (\vec{m}_A, \vec{d}_A); (\vec{m}_B, \vec{d}_B) \right)\right]$. One may also consider corresponding constraints in the situation when the input $i_A^{(k)}$ is allowed to influence the outputs $o_B^{(k+1)}, \dots, o_B^{(N)}$ for any $1 \leq k \leq N$, and vice versa. In this scenario, the spatial separation between the parties only restricts communication within each round, while input choices of any one run are able to influence the outputs of either party in subsequent runs.  The corresponding extension of the Time-Ordered No-Signaling polytope in that scenario is defined by fewer restrictions than the $\textbf{TONS}\left[\textbf{B}\left(2; (\vec{m}_A, \vec{d}_A); (\vec{m}_B, \vec{d}_B) \right)\right]$ polytope considered in this paper. 

The boxes within the Time-Ordered No-Signaling polytope that in addition satisfy the integrality constraint given by 
\begin{enumerate}
\setcounter{enumi}{3}
\item Integrality $P_{\textbf{O}_A, \textbf{O}_B | \textbf{I}_A, \textbf{I}_B}(\textbf{o}_A, \textbf{o}_B | \textbf{i}_A, \textbf{i}_B) \in  \{0,1\}$ for all $\textbf{o}_A, \textbf{o}_B, \textbf{i}_A, \textbf{i}_B$,
\end{enumerate}
are said to be Time-Ordered Local Deterministic boxes. The convex hull of these boxes forms the Time-Ordered Local polytope denoted by $\textbf{TOLoc}\left[\textbf{B}\left(2; (\vec{m}_A, \vec{d}_A); (\vec{m}_B, \vec{d}_B) \right)\right]$ \cite{GWCAN}. This is the set of correlations obtainable in time-ordered local models, i.e., where the response of each party for its $j$-th measurement depends only on a hidden variable $\lambda$, the first $j$ measurement settings and the first $j-1$ measurement outcomes of that party.  

The set of sequential Quantum Correlations denoted by $\textbf{Qseq}\left[\textbf{B}\left(2; (\vec{m}_A, \vec{d}_A); (\vec{m}_B, \vec{d}_B) \right)\right]$ also lies within the Time-Ordered No-Signaling polytope. This set consists of boxes $P$ where each component $P_{\textbf{O}_A, \textbf{O}_B | \textbf{I}_A, \textbf{I}_B}(\textbf{o}_A, \textbf{o}_B | \textbf{i}_A, \textbf{i}_B)$ is obtained as
\begin{eqnarray}
P_{\textbf{O}_A, \textbf{O}_B | \textbf{I}_A, \textbf{I}_B}(\textbf{o}_A, \textbf{o}_B | \textbf{i}_A, \textbf{i}_B) = \langle \psi | \textbf{E}^A_{\textbf{i}_A, \textbf{o}_A} \otimes \textbf{E}^B_{\textbf{i}_B, \textbf{o}_B} | \psi \rangle. \nonumber \\
\end{eqnarray}
Here $|\psi \rangle \in \mathds{C}^{d_{seq}}$ for some arbitrary dimension $d_{seq}$, and $\textbf{E}^A_{\textbf{i}_A, \textbf{o}_A}, \textbf{E}^B_{\textbf{i}_B, \textbf{o}_B}$ are projective measurements defined by \cite{BBS}
\begin{widetext}
\begin{eqnarray}
\textbf{E}^J_{\textbf{i}_J, \textbf{o}_J} := \sum_{\alpha^{(1)}_J, \dots, \alpha^{(N)}_J} \left( K^{J}_{i^{(1)}_J, o^{(1)}_J, \alpha^{(1)}_J} \right)^{\dagger} \dots  \left( K^{J}_{i^{(N)}_J, o^{(N)}_J, \alpha^{(N)}_J} \right)^{\dagger} K^{J}_{i^{(N)}_J, o^{(N)}_J, \alpha^{(N)}_J} \dots K^{J}_{i^{(1)}_J, o^{(1)}_J, \alpha^{(1)}_J}, 
\end{eqnarray}
\end{widetext}
for $J = A, B$ with sets of Kraus operators $\bigg\{K^{J}_{i^{(j)}_J, o^{(j)}_J, \alpha^{(j)}_J} \bigg\}$ corresponding to the $j$-th run of party $J$ satisfying
\begin{eqnarray}
\sum_{o^{(j)}_J, \alpha^{(j)}_J} \left( K^{J}_{i^{(j)}_J, o^{(j)}_J, \alpha^{(j)}_J} \right)^{\dagger} K^{J}_{i^{(j)}_J, o^{(j)}_J, \alpha^{(j)}_J} = \mathbf{1}, \; \; \forall i^{(j)}_J.
\end{eqnarray}
Here, the sum over $\alpha^{(j)}_J$ is incorporated to allow for the fact that multiple Kraus operators may be mapped to a single output $o^{(j)}_J$ \cite{BBS}. Being projective, the measurement operators can be shown to satisfy 
\begin{equation}
\label{eq:Proj}
\textbf{E}^J_{\textbf{i}_J, \textbf{o}_J} \textbf{E}^J_{\textbf{i}_J, \textbf{o'}_J} = \delta_{\textbf{o}_J, \textbf{o'}_J} \textbf{E}^J_{\textbf{i}_J, \textbf{o}_J} \; \forall \textbf{i}_J, \textbf{o}_J, \textbf{o'}_J
\end{equation}
They also satisfy the condition of sequential measurement: 
\begin{eqnarray}
\label{eq:SM2}
\sum_{o^{(k)}_J, \dots, o^{(N)}_J}  \textbf{E}^J_{\textbf{i}_J, \textbf{o}_J} = \sum_{o^{(k)}_J, \dots, o^{(N)}_J}  \textbf{E}^J_{\textbf{i'}_J, \textbf{o}_J} \nonumber \\ \forall 2 \leq k \leq N, o^{(1)}_J, \dots, o^{(k-1)}_J, \nonumber \\
 \forall \textbf{i}_J, \textbf{i'}_J \; \text{s.t.} \; i^{(j)}_J = i'^{(j)}_J \; (1 \leq j \leq k-1)
\end{eqnarray}
The projective measurement operators also satisfy a condition analogous to the 'Local Orthogonality' condition \cite{LO}:
\begin{eqnarray}
\label{eq:LO}
&&\textbf{E}^J_{\textbf{i}_J, \textbf{o}_J} \textbf{E}^J_{\textbf{i'}_J, \textbf{o'}_J} = 0 \; \forall \textbf{i}_J, \textbf{o}_J, \textbf{i'}_J, \textbf{o'}_J \; \text{s.t.} \nonumber \\
&&(i^{(1)}_J, \dots, i^{(k)}_J) = (i'^{(1)}_J, \dots, i'^{(k)}_J) \;  \wedge \nonumber \\ &&(o^{(1)}_J, \dots, o^{(k)}_J) \neq (o'^{(1)}_J, \dots, o'^{(k)}_J) \; \text{for} \; 1 \leq k \leq N. \nonumber \\
\end{eqnarray}

The set $\textbf{Qseq}\left[\textbf{B}\left(2; (\vec{m}_A, \vec{d}_A); (\vec{m}_B, \vec{d}_B) \right)\right]$ is convex but not a polytope. We have the inclusions $\textbf{TOLoc}\left[\textbf{B}\left(2; (\vec{m}_A, \vec{d}_A); (\vec{m}_B, \vec{d}_B) \right)\right] \subseteq \textbf{Qseq}\left[\textbf{B}\left(2; (\vec{m}_A, \vec{d}_A); (\vec{m}_B, \vec{d}_B) \right)\right] \subseteq \textbf{TONS}\left[\textbf{B}\left(2; (\vec{m}_A, \vec{d}_A); (\vec{m}_B, \vec{d}_B) \right)\right]$.

The general hierarchy of semi-definite programs for polynomial optimization with non-commuting variables was introduced by Pironio, Navascu\'{e}s and Ac\'{i}n in \cite{NPA2}, which in principle covers also the scenario of sequential Bell non-locality. Nevertheless, this specific case was explicitly handled recently in \cite{BBS} whose treatment we may follow. The hierarchy was defined in an exactly analogous fashion to the NPA hierarchy for single-run Bell experiments treating the sequences of measurements as single measurements, i.e., based on the projection operators $\textbf{E}^A_{\textbf{i}_A, \textbf{o}_A}, \textbf{E}^B_{\textbf{i}_B, \textbf{o}_B}$ in place of the operators $E^{A}_{i_A, o_A}  E^{B}_{i_B, o_B}$. The additional structure of sequential measurements is taken care of by including linear constraints (\ref{eq:Proj}), (\ref{eq:SM2}), (\ref{eq:LO}) in the certificate. The hierarchy was in this case shown to converge to the set $\textbf{Qseq}^{pr}\left[\textbf{B}\left(2; (\vec{m}_A, \vec{d}_A); (\vec{m}_B, \vec{d}_B) \right)\right]$ consisting of boxes $P$ where each component $P_{\textbf{O}_A, \textbf{O}_B | \textbf{I}_A, \textbf{I}_B}(\textbf{o}_A, \textbf{o}_B | \textbf{i}_A, \textbf{i}_B)$ is obtained as 
\begin{eqnarray}
P_{\textbf{O}_A, \textbf{O}_B | \textbf{I}_A, \textbf{I}_B}(\textbf{o}_A, \textbf{o}_B | \textbf{i}_A, \textbf{i}_B) = \langle \psi | \textbf{E}^A_{\textbf{i}_A, \textbf{o}_A}  \textbf{E}^B_{\textbf{i}_B, \textbf{o}_B} | \psi \rangle, \nonumber \\
\end{eqnarray}
where the measurements of the different parties commute, i.e., $\left[ \textbf{E}^A_{\textbf{i}_A, \textbf{o}_A}, \textbf{E}^B_{\textbf{i}_B, \textbf{o}_B} \right] = 0$ for all $\textbf{i}_A, \textbf{o}_A, \textbf{i}_B, \textbf{o}_B$.

The Time-Ordered No-Signaling polytope consisting of boxes obeying the non-negativity, normalization and time-ordered no-signaling constraints can be expressed succintly as
\begin{eqnarray}
\label{eq:TONS-def}
\textbf{TONS}\left[\textbf{B}\left(2; (\vec{m}_A, \vec{d}_A); (\vec{m}_B, \vec{d}_B) \right)\right] = \nonumber \\ 
\big\{ | P \rangle \in \mathds{R}^{n_{seq}} : A \cdot | P \rangle \leq | b \rangle \big\}.
\end{eqnarray} 
Here, the matrix $A$ and vector $|b \rangle$ encode the non-negativity, normalization and time-ordered no-signaling constraints, and the box $P$ is written as a vector $| P \rangle$ of length $n_{seq} = \left(m d \right)^{2N}$. 

We are interested in the question whether quantum correlations can realize the extremal boxes of the Time-Ordered No-Signaling polytope, where an extremal box or a vertex is one that cannot be expressed as a non-trivial convex combination of boxes in the TONS polytope. An equivalent mathematical characterization of vertices is based on the following fact:
\begin{fact}
A box $P$ is a vertex of the Time-Ordered No-Signaling polytope $\textbf{TONS}\left[\textbf{B}\left(2; (\vec{m}_A, \vec{d}_A); (\vec{m}_B, \vec{d}_B) \right)\right]$ if and only if $rank(\tilde{A}_P) = n_{seq}$, where $\tilde{A}_P$ denotes the sub-matrix of the matrix $A$ from (\ref{eq:TONS-def}) consisting of those row vectors $A_i$ for which $A_i \cdot | P \rangle = |b \rangle_i$. 
\end{fact}
In other words, every extremal box $P$ satisfies, besides the normalization and time-ordered no-signaling equality constraints, a certain number of the non-negativity inequalities with equality, i.e., the box sets $P_{\textbf{O}_A, \textbf{O}_B | \textbf{I}_A, \textbf{I}_B}(\textbf{o}_A, \textbf{o}_B | \textbf{i}_A, \textbf{i}_B) = 0$ for a uniquely identifiable set of $\textbf{o}_A, \textbf{i}_A, \textbf{o}_B, \textbf{i}_B$. For two extremal boxes $P$ and $P'$, the corresponding sub-matrices are not equal $\tilde{A}_P \neq \tilde{A}_{P'}$ and the sub-matrix $\tilde{A}_P$ can therefore be used to uniquely identify the vertex $P$. Finally, a local vertex of the Time-Ordered No-Signaling polytope $\textbf{TONS}\left[\textbf{B}\left(2; (\vec{m}_A, \vec{d}_A); (\vec{m}_B, \vec{d}_B) \right)\right]$ is one that has only entries in $\{0,1\}$ by the integrality constraint stated earlier. Therefore, a non-local vertex is one that has at least one non-integral entry, i.e., one entry that is neither $0$ nor $1$. 

We now state the central result of this section as the following theorem.

\begin{theorem}
For any $(\vec{m}_A, \vec{d}_A), (\vec{m}_B, \vec{d}_B)$ let $P$ be an extremal box of the Time-Ordered No-Signaling polytope $\textbf{TONS}\left[\textbf{B}\left(2; (\vec{m}_A, \vec{d}_A); (\vec{m}_B, \vec{d}_B) \right)\right]$ such that $P \notin \textbf{TOLoc}\left[\textbf{B}\left(2; (\vec{m}_A, \vec{d}_A); (\vec{m}_B, \vec{d}_B) \right)\right]$. Then, $P \notin \text{cl}\left(\textbf{Qseq}\left[\textbf{B}\left(2; (\vec{m}_A, \vec{d}_A); (\vec{m}_B, \vec{d}_B) \right)\right] \right)$. 
\end{theorem}

\begin{proof}
Following the Cabello-Severini-Winter (CSW) framework for single-run Bell non-locality \cite{CSW1, CSW2}, we define an orthogonality graph $G_{\textbf{B}}$ corresponding to a sequential Bell scenario as follows. To each measurement event $\left(\textbf{o}_A, \textbf{o}_B | \textbf{i}_A, \textbf{i}_B \right)$ of the Bell scenario, we associate a vertex $v_{\left(\textbf{o}_A, \textbf{o}_B | \textbf{i}_A, \textbf{i}_B \right)}$ of the graph $G_{\textbf{B}}$.  Two such vertices $v_{\left(\textbf{o}_A, \textbf{o}_B | \textbf{i}_A, \textbf{i}_B \right)}, v_{\left(\textbf{o'}_A, \textbf{o'}_B | \textbf{i'}_A, \textbf{i'}_B \right)}$ are connected by an edge if and only if the associated measurements are orthogonal, 
\begin{equation}
\label{eq:LO2}
\textbf{E}^A_{\textbf{i}_A, \textbf{o}_{A}} \textbf{E}^B_{\textbf{i}_B, \textbf{o}_{B}} \textbf{E}^A_{\textbf{i'}_A, \textbf{o'}_{A}} \textbf{E}^B_{\textbf{i'}_B, \textbf{o'}_{B}} = 0,
\end{equation}
i.e., either $\textbf{E}^A_{\textbf{i}_A, \textbf{o}_A} \textbf{E}^A_{\textbf{i'}_A, \textbf{o'}_A} = 0$ or $ \textbf{E}^B_{\textbf{i}_B, \textbf{o}_B} \textbf{E}^B_{\textbf{i'}_B, \textbf{o'}_B} = 0$. 

A crucial observation is that the normalization and time-ordered no-signaling constraints on a box $P$ are equivalent to maximum clique (in)equalities of the graph $G_{\textbf{B}}$. Here, a clique inequality is an inequality of the form $\sum_{i \in c} | P \rangle_i \leq 1$ for some clique $c$ in the graph (a clique is a set of mutually edge-connected vertices). The fact that normalization is a clique inequality is clear, since normalization by definition considers events with the same inputs $\textbf{i}_A, \textbf{i}_B$ and different outputs $\textbf{o}_A, \textbf{o}_B$. That the time-ordered no-signaling constraints also form a clique inequality can be seen by a modification of the argument from \cite{LO} as:
\begin{widetext}
\begin{eqnarray}
\label{eq:tons-as-lo}
\sum_{o_A^{(k)}, \dots, o_A^{(N)}} P_{\textbf{O}_A, \textbf{O}_B | \textbf{I}_A, \textbf{I}_B}(\textbf{o}_A, \textbf{o}_B | \textbf{i}_A, \textbf{i}_B) + \sum_{(o'^{(1)}_A, \dots, o'^{(k-1)}_A) \neq (o^{(1)}_A, \dots, o^{(k-1)}_A) } \sum_{o'^{(k)}_A, \dots, o'^{(N)}_A} P_{\textbf{O}_A, \textbf{O}_B | \textbf{I}_A, \textbf{I}_B}(\textbf{o'}_A, \textbf{o}_B | \textbf{i'}_A, \textbf{i}_B)  \nonumber \\ + \sum_{\textbf{o''}_B \neq \textbf{o}_B} \sum_{\textbf{o''}_A} P_{\textbf{O}_A, \textbf{O}_B | \textbf{I}_A, \textbf{I}_B}(\textbf{o''}_A, \textbf{o''}_B | \textbf{i''}_A, \textbf{i}_B) = 1, \; \nonumber \\
\forall \; 1 \leq k \leq N, \textbf{i}_A, \textbf{i'}_A \; \text{with} \; (i^{(1)}_A, \dots, i^{(k-1)}_A) = (i'^{(1)}_A, \dots, i'^{(k-1)}_A), \; \forall \textbf{i''}_A, \textbf{i}_B, \textbf{o}_A, \textbf{o}_B \nonumber \\
\end{eqnarray}
\end{widetext}
A similar equation shows that $\sum_{o_B^{(k)}, \dots, o_B^{(N)}} P_{\textbf{O}_A, \textbf{O}_B | \textbf{I}_A, \textbf{I}_B}(\textbf{o}_A, \textbf{o}_B | \textbf{i}_A, \textbf{i}_B)$ is independent of  $i_B^{(k)}, \dots, i_B^{(N)}$ for all $1 \leq k \leq N, \textbf{i}_A, \textbf{o}_A, \textbf{i}_B, \textbf{o}_B$. It is readily seen that each of the measurement events in (\ref{eq:tons-as-lo}) correspond to orthogonal measurements (\ref{eq:LO2}). Therefore, by the definition of the edge in (\ref{eq:LO2}) the time-ordered no-signaling constraints correspond to a saturated clique inequality. The dimension of the convex sets in question is given by $n_{seq} - D$ where $D$ denotes the number of independent normalization and time-ordered no-signaling constraints. 

Now, in the hierarchy for sequential Bell non-locality, we consider sets of sequences of product projection operators $S^{seq}_1 = \{ \mathds{1} \} \cup \{\textbf{E}^{A}_{\textbf{i}_A, \textbf{o}_A} \} \cup \{\textbf{E}^{B}_{\textbf{i}_B, \textbf{o}_B}\}$, $S^{seq}_2 = S_1 \cup \{\textbf{E}^{A}_{\textbf{i}_A, \textbf{o}_A} \textbf{E}^{B}_{\textbf{i}_B, \textbf{o}_B}\}$, etc. The convex sets $\textbf{Qseq}_l\left[\textbf{B}\left(2; (\vec{m}_A, \vec{d}_A); (\vec{m}_B, \vec{d}_B) \right)\right]$ corresponding to different levels of this hierarchy consist of boxes $P$ associated with a positive semi-definite certificate $\Gamma^{seq, l}$. The certificate $\Gamma^{seq, l}$ is associated to the set of operators $S^{seq}_l$ and is a $\big|S^{seq}_l \big| \times \big|S^{seq}_l \big|$ matrix indexed by the operators in $S^{seq}_l$. The certificate is a Hermitian positive semi-definite matrix satisfying (i) $\Gamma^{seq, l}_{\mathds{1}, \mathds{1}} = 1$, (ii) $\Gamma^{seq, l}_{Q, R} = \Gamma^{seq, l}_{S, T}$ if $Q^{\dagger} R = S^{\dagger} T$. The latter condition in particular imposes that $\Gamma^{seq, l}_{\mathds{1}, \textbf{E}^{A}_{\textbf{i}_A, \textbf{o}_A} \textbf{E}^{B}_{\textbf{i}_B, \textbf{o}_B}} = \Gamma^{seq, l}_{\textbf{E}^{A}_{\textbf{i}_A, \textbf{o}_A}, \textbf{E}^{B}_{\textbf{i}_B, \textbf{o}_B}} = \Gamma^{seq, l}_{\textbf{E}^{A}_{\textbf{i}_A, \textbf{o}_A} \textbf{E}^{B}_{\textbf{i}_B, \textbf{o}_B}, \textbf{E}^{A}_{\textbf{i}_A, \textbf{o}_A} \textbf{E}^{B}_{\textbf{i}_B, \textbf{o}_B}} = P_{\textbf{O}_A, \textbf{O}_B | \textbf{I}_A, \textbf{I}_B}(\textbf{o}_A, \textbf{o}_B | \textbf{i}_A, \textbf{i}_B)$. The lack of existence of a positive semi-definite certificate corresponding to any level of the hierarchy for a box $P$ implies the exclusion of that box from the quantum set. Following earlier investigations in \cite{our}, we consider a level of the hierarchy that we denote $\tilde{\textbf{Q}}_{seq}\left[\textbf{B}\left(2; (\vec{m}_A, \vec{d}_A); (\vec{m}_B, \vec{d}_B) \right)\right]$ corresponding to the set of operators $\{ \mathds{1} \} \cup \{\textbf{E}^{A}_{\textbf{i}_A, \textbf{o}_A} \textbf{E}^{B}_{\textbf{i}_B, \textbf{o}_B}\}$. This level is the analog of the Almost Quantum set for bipartite correlations in the traditional Bell scenario (see \cite{AQ}). Crucially, the only constraints from this set of operators are those that impose the Hermiticity of the matrix by the identity
\begin{eqnarray}
\left( \textbf{E}^{A}_{\textbf{i}_A, \textbf{o}_A} \textbf{E}^{B}_{\textbf{i}_B, \textbf{o}_B} \right) \left(\textbf{E}^{A}_{\textbf{i'}_A, \textbf{o'}_A} \textbf{E}^{B}_{\textbf{i'}_B, \textbf{o'}_B} \right)^{\dagger} = \nonumber \\ \left( \left( \textbf{E}^{A}_{\textbf{i'}_A, \textbf{o'}_A} \textbf{E}^{B}_{\textbf{i'}_B, \textbf{o'}_B} \right) \left(\textbf{E}^{A}_{\textbf{i}_A, \textbf{o}_A} \textbf{E}^{B}_{\textbf{i}_B, \textbf{o}_B} \right)^{\dagger} \right)^{\dagger},
\end{eqnarray}
due to the commutativity of the operators corresponding to different parties. As in previous studies for the single-run Bell scenario \cite{our, Fritz2, AQ}, this set may be identified with the familiar Lov\'{a}sz-theta set $TH(G_{\textbf{B}})$ from graph theory, with the clique inequalities corresponding to the normalization and the time-ordered no-signaling conditions (\ref{eq:tons-as-lo}) being set to equalities. In other words, $\tilde{\textbf{Q}}_{seq}\left[\textbf{B}\left(2; (\vec{m}_A, \vec{d}_A); (\vec{m}_B, \vec{d}_B) \right)\right] = TH(G_{\textbf{B}}) \cap C_{n, tons}$ where $C_{n, tons}$ denotes the set of clique equalities from the normalization and the time-ordered no-signaling constraints. To see this explicitly, note that the set $TH(G_{\textbf{B}})$ is defined as \cite{Schrijver}
\begin{widetext}
\begin{equation}
\label{theta-body-2}
TH(G_{\textbf{B}}) = \left\{ 
\begin{tabular}{ c|c c } 
$| \mathcal{P} \rangle \in \mathbb{R}^{|V|}$ & $\exists \mathbf{\Pi} \in \mathbb{S}^{|V|} \; \; \text{s.t.}$ & \;\;\;\; $\mathbf{\Pi}_{i,j} = 0 \quad ((i,j) \in E),$ \\
   &&$\mathbf{\Pi}_{i,i} = | \mathcal{P} \rangle_i \quad (i \in V),$ \\
   &&$\mathbf{\Pi} - | \mathcal{P} \rangle \langle \mathcal{P} | \succeq 0$\\
 \end{tabular} 
 \right\}
\end{equation}
\end{widetext}
Here $V$ denotes the vertex set of the graph $G_{\textbf{B}}$ and $\mathbb{S}^{|V|}$ denotes the set of symmetric matrices of size $|V| \times |V|$. 
We may identify the requirement on the certificate of $\tilde{\textbf{Q}}_{seq}$ with the constraints $C_{n, tons}$ together with the requirement of the real symmetric positive semi-definite matrix $\begin{pmatrix}
\mathds{1}& \langle P | \\ | P \rangle & \mathbf{\Pi}
\end{pmatrix}$ which by the use of Schur complements is equivalent to the condition that $\mathbf{\Pi} -  | \mathcal{P} \rangle \langle \mathcal{P} | \succeq 0$ in Eq.(\ref{theta-body-2}). 

Furthermore, we have that $\text{cl}\left(\textbf{Qseq}\left[\textbf{B}\left(2; (\vec{m}_A, \vec{d}_A); (\vec{m}_B, \vec{d}_B) \right)\right] \right) \subseteq \tilde{\textbf{Q}}_{seq}\left[\textbf{B}\left(2; (\vec{m}_A, \vec{d}_A); (\vec{m}_B, \vec{d}_B) \right)\right]$ since the sequential hierarchy converges to the closure of the corresponding quantum set, so that exclusion from $\tilde{\textbf{Q}}_{seq}$ implies exclusion from $\text{cl}\left(\textbf{Qseq}\right)$. Putting these facts together, and following analogous steps to \cite{our, FT02, Shepherd, Silva} now gives the result. 
\end{proof}

\section{Supplemental Material: Quantum realisation of extremal non-signaling non-local assemblages}
Here we give a formal proof of the propositions from the main text that state that quantum theory allows for the realisation of extremal non-signalling non-local assemblages. 

\subsection{Postquantum steering}

In this subsection we will recall the notions already introduced in the main text and for the reader's convenience we will evoke explicit reasoning backing up some statements presented in the main text of the paper.

We will begin by considering so called steering scenario, firstly proposed by Schr{\"o}dinger \cite{S36} and putted into modern perspective in \cite{WJD07}. Assume that two distant (separated) subsystems A (Alice) and B (Bob) share a quantum state. Bob's subsystem is completely characterised, dimension $d$ of its underlying Hilbert space is known as well as all potential operations which may be performed on it by Bob. Alice may perform measurements with settings labeled by $x$ (from some finite and fixed set of indexes $\mathcal{X}$) and outcomes labeled by $a$ from some finite and fixed set of indexes $\mathcal{A}$ (note that in full generality $\mathcal{A}_x$ may be dependent on the choice of $x$, for simplicity we will omit this possibility), but apart form that, subsystem A in uncharacterised. Dimension of Hilbert space describing subsystem A is unknown and measurements performed on subsystem A, described by POVM elements $M_{a|x}$, may be arbitrary. After obtaining $a$ while measuring $x$, local state of subsystem B is characterised by a subnormalised state
\begin{equation}\label{first}
\sigma_{a|x}=\mathrm{Tr}_{A}(M^{(A)}_{a|x}\otimes \mathds{1}\rho^{(AB)}),
\end{equation}i.e. while measuring $x$ Alice obtains $a$ with probability $\mathrm{Tr}_{B}(\sigma_{a|x})$, and with that probability Bob's subsystem is described by the normalised version of $\sigma^{(B)}_{a|x}$. Described scenario can be therefore characterised by the collection of subnormalised states $\Sigma^{(B)}=\left\{\sigma^{(B)}_{a|x}\right\}_{a,x}$ acting on $d_B$ dimensional space. This family is known as an (quantum) assemblage. Observe that (\ref{first}) induces no-signaling conditions (together with positivity and normalisation)
\begin{equation}\label{sup_no1}
\forall_{a,x}\ \sigma^{(B)}_{a|x}\geq 0,
\end{equation}
\begin{equation}\label{sup_no2}
\forall_{x,x'}\ \sum_a\sigma^{(B)}_{a|x}=\sigma^{(B)}=\sum_a\sigma^{(B)}_{a|x'},
\end{equation}
\begin{equation}\label{sup_no3}
\mathrm{Tr}(\sigma^{(B)})=1,
\end{equation}where $\sigma^{(B)}$ is a local state of B. Now consider special case of a quantum steering in which the shared state $\tilde{\rho}^{(AB)}=\sum_i q_i \rho^{(A)}_i\otimes \sigma^{(B)}_i$ is separable ($\rho^{(A)}_i, \sigma^{(B)}_i$ denote local states and $q_i\geq 0, \sum_i q_i=1$). Note that in such a case, according to (\ref{first}), we obtain
\begin{equation}\label{sup_lhs}
\sigma^{(B)}_{a|x}=\sum_i q_i\mathrm{Tr}(M^{(A)}_{a|x}\rho_i^{(A)})\sigma_i^{(B)}=\sum_i q_ip^{(A)}_i(a|x)\sigma^{(B)}_i
\end{equation}where we use the fact that $\mathrm{Tr}(M^{(A)}_{a|x}\rho_i)=p^{(A)}_i(a|x)$ defines a conditional probability distribution. We say that assemblage of the form (\ref{sup_lhs}) admits a LHS (local hidden state) model, or simply that it is LHS. If $\tilde{p}^{(A)}_k(a|x)$ denote deterministic conditional probability distributions, then we may use the convex combination $p^{(A)}_i(a|x)=\sum_{k}s_{ki}\tilde{p}^{(A)}_k(a|x)$ and write (\ref{sup_lhs}) as
\begin{equation}\label{sup_lhs22}
\sigma_{a|x}=\sum_{ik} q_is_{ki}\tilde{p}^{(A)}_k(a|x)\sigma^{(B)}_i=\sum_{k} \tilde{q}_k\tilde{p}^{(A)}_k(a|x)\tilde{\sigma}^{(B)}_k
\end{equation}with states $\tilde{\sigma}^{(B)}_{k}=\frac{1}{\tilde{q}_{k}}\sum_i q_is_{ki}\sigma^{(B)}_i$ and $\tilde{q}_{k}=\mathrm{Tr}(\sum_i q_is_{ki}\sigma^{(B)}_i)$ as coefficients of convex combinations. Observe that assemblage is LHS if and only if it has a quantum realisation with a bipartite separable state. Thus distinction of LHS assembles mirrors existence of separable states in a bipartite setting.

Finally, one can consider situation in which system B is described by a local quantum state, but the composite system is characterised rather by some no-signaling generalised probabilistic theory. While Alice measures $x$ she still obtains $a$ with probability $\mathrm{Tr}_{B}(\sigma^{(B)}_{a|x})$, and with that probability Bob's subsystem is described by the normalised version of $\sigma^{(B)}_{a|x}$, but the measurements are no longer expressed in the language of quantum theory and $\sigma^{(B)}_{a|x}$ may not be described by the formula (\ref{first}). Measurements are now seen as a device with classical inputs $x$ and classical outputs $a$, with no-signaling constraints expressed by the relation between given operators $\sigma^{(B)}_{a|x}$. Abstract no-signaling assemblage is defined as a collection of operators $\Sigma^{(B)}=\left\{\sigma_{a|x}^{(B)}\right\}_{a,x}$ acting on fixed $d_B$ dimensional space, which satisfy constraints (\ref{sup_no1}-\ref{sup_no3}) \cite{SBCSV15}. In principal we say that a no-signaling assemblage $\Sigma^{(B)}$ admits quantum realisation (or simply that $\Sigma^{(B)}$ is quantum assemblage) if there is a choice of subsystem A, POVMs and a joint state of two subsystems such that $\Sigma^{(B)}$ can be expressed by the formula (\ref{first}). Similarly, we may talk about no-signaling assemblages with LHS model.

Note that the set of no-signaling assemblages is convex. Indeed, if $\pi^{(B)}_{a|x}, \tau^{(B)}_{a|x}$ define no-signaling assemblages, then $\sigma^{(B)}_{a|x}=p\pi^{(B)}_{a|x}+(1-p)\tau^{(B)}_{a|x}$ for any $0\leq p\leq 1$ fulfills conditions (\ref{sup_no1}-\ref{sup_no3}), so it defines a no-signaling assemblage as well. Similar observation holds also for quantum assemblages and LHS assemblages. To see this assume that there are two quantum assembles given by $\pi^{(B)}_{a|x}=\mathrm{Tr}_{A}(M^{(A)}_{a|x}\otimes \mathds{1}\rho^{(AB)}_1)$ and $\tau^{(B)}_{a|x}=\mathrm{Tr}_{A'}(N^{(A')}_{a|x}\otimes \mathds{1}\rho_2^{(A'B)})$ with $\rho_1^{(AB)}$ being a state of system AB and $\rho_2^{(A'B)}$ being a state of system A'B. Then
\begin{equation}\nonumber
\sigma^{(B)}_{a|x}=p\pi^{(B)}_{a|x}+(1-p)\tau^{(B)}_{a|x}=\mathrm{Tr}_{A\oplus A'}(O^{(A\oplus A')}_{a|x}\otimes \mathds{1}\rho^{((A\oplus A')B)})
\end{equation}where $\rho^{((A\oplus A')B)}=p\rho^{(AB)}_1\oplus (1-p)\rho^{(A'B)}_2$ and $O^{(A\oplus A')}_{a|x}=M^{(A)}_{a|x}\oplus N^{(A')}_{a|x}$, so $\sigma^{(B)}_{a|x}$ defines quantum assemblage. Note that $A\oplus A'$ denotes here composed system described by the direct sum of a matrix algebras describing subsystem A and A' respectively. Similarly expression $(A\oplus A')B$ refers to composed system described by the tensor product of the previously evoked algebra and the matrix algebra which describes subsystem B (therefore, from this perspective both $\rho^{(AB)}_1$ and $\rho^{(A'B)}_2$ can be seen as states on larger system $(A\oplus A')B$).

Finally, convexity of the set of LHS assemblages follows form characterisation (\ref{sup_lhs22}).

It is however well known \cite{G89,HJW93}, that any abstract no-signaling assemblage (in the discussed bipartite scenario) admits a quantum realisation, i.e. there exist subsystem A described by some Hilbert space, POVMs elements $M^{(A)}_{a|x}$ acting on that space and a joint quantum state $\rho^{(AB)}$ of subsystems A and B, such that $\sigma^{(B)}_{a|x}$ is reconstructed by formula (\ref{first}). Indeed, for a given no-signaling assemblage  $\Sigma^{(B)}=\left\{\sigma_{a|x}\right\}_{a,x}$ assume spectral decomposition $\sum_a\sigma^{(B)}_{a|x}=\sigma^{(B)}=\sum_k \lambda_k |\phi^{(B)}_k\rangle \langle \phi^{(B)}_k|$ and put $|\psi^{(AB)}\rangle=\sqrt{\lambda_k}|\phi_k^{(A)}\rangle|\phi_k^{(B)}\rangle$. Let $M^{(A)}_{a|x}=(\sigma^{(A)})^{-\frac{1}{2}}(\sigma^{(A)}_{a|x})^T(\sigma^{(A)})^{-\frac{1}{2}}$ for $a\neq |\mathcal{A}|-1$ and
$M^{(A)}_{a|x}=(\sigma^{(A)})^{-\frac{1}{2}}(\sigma^{(A)}_{a|x})^T(\sigma^{(A)})^{-\frac{1}{2}}+(\mathds{1}-R_E)$ for $a= |\mathcal{A}|-1$ with $(\sigma^{(A)})^{-\frac{1}{2}}=\sum_k \lambda_k^{-\frac{1}{2}} |\phi_k^{(A)}\rangle \langle \phi_k^{(A)}|$, $R_E$ being projection on the image of  $\sigma^{(B)}$ and $T$ denoting transposition (with respect to eigenvectors of $\sigma^{(B)}$). Direct calculations show that with this characterisation $\sigma^{(B)}_{a|x}=\mathrm{Tr}_{A}(M^{(A)}_{a|x}\otimes \mathds{1}|\psi^{(AB)}\rangle \langle \psi^{(AB)}|)$. Therefore, we may say that there is no postquantum steering (any no-signaling assemblage is a quantum assemblage).

In order to witness postquantum steering, we will consider the simplest generalisation of bipartite no-signaling steering scenario. Let us restrict to the tripartite case with two uncharacterised subsystems. We consider three distant (separated) subsystems A (Alice), B (Bob), C (Charlie). While system C has quantum nature and it is described by some Hilbert space of known dimension $d_C$, subsystems of A and B remains unknown to Charlie and can be seen as devices described only by set of classical inputs (labels of measurements settings) and classical outputs (labels of measurements outcomes). Once more Alice and Bob may perform local (and independent) measurements with finite settings labeled by $x\in \mathcal{X}$ or $y\in\mathcal{Y}$ (respectively) and finite outcomes labeled by $a\in\mathcal{A}$ or $b\in\mathcal{B}$  (respectively). According to measurements with fixed settings $x,y$, quantum system of Charlie is described by the normalised version of operator $\sigma^{(C)}_{ab|xy}$ with probability $\mathrm{Tr}(\sigma^{(C)}_{ab|xy})$ (which also describes the probability of Alice and Bob obtaining outcomes labeled by $a,b$). The fact that measurements of Alice and Bob are local (do not affect each other) is encoded by the following no-signaling conditions (together with positivity and normalisation) imposed on operators $\sigma^{(C)}_{ab|xy}$
\begin{equation}\label{sup_as1}
\forall_{a,b,x,y}\ \sigma^{(C)}_{ab|xy}\geq 0,
\end{equation}
\begin{equation}\label{sup_as2}
\forall_{b,x,x',y}\ \sum_a\sigma^{(C)}_{ab|xy}=\sum_a\sigma^{(C)}_{ab|x'y},
\end{equation}
\begin{equation}\label{sup_as3}
\forall_{a,x,y,y'}\ \sum_b\sigma^{(C)}_{ab|xy}=\sum_b\sigma^{(C)}_{ab|xy'},
\end{equation}
\begin{equation}\label{sup_as4}
\forall_{x,y}\ \mathrm{Tr}(\sum_{a,b}\sigma^{(C)}_{ab|xy})=\mathrm{Tr}(\sigma^{(C)})=1, 
\end{equation}where $\sigma^{(C)}$ is interpreted as a local state of subsystem C. In analogy to the previously discussed case, abstract no-signaling assemblage (for tripartite steering with two uncharacterised subsystems) is defined as a collection of operators $\Sigma^{(C)}=\left\{\sigma^{(C)}_{ab|xy}\right\}_{a,b,x,y}$ acting on fixed $d_C$ dimensional space, which satisfy constraints (\ref{sup_as1}-\ref{sup_as4}) \cite{SBCSV15}.  
If for given no-signaling assemblage $\left\{\sigma^{(C)}_{ab|xy}\right\}_{a,b,x,y}$ there exist subsystems A and B, local POVMs given respectively by elements $M^{(A)}_{a|x},N^{(B)}_{b|y}$ and some tripartite state $\rho^{(ABC)}$ of the composed system ABC (note that subsystem C is always fixed for considered set of no-signaling assemblages, as fixed is dimension $d_C$) such that
\begin{equation}\label{sup_q2}
\sigma^{(C)}_{ab|xy}=\mathrm{Tr}_{AB}(M^{(A)}_{a|x}\otimes N^{(B)}_{b|y}\otimes \mathds{1}\rho^{(ABC)}),
\end{equation}we say that $\Sigma^{(C)}$ admits quantum realisation (or simply that $\Sigma^{(C)}$ is a quantum assemblage). However, in contrast with bipartite steering, it appears that not all no-signaling assemblages in tripartite scenario admits quantum realisation (i.e. we say that there are postquantum assemblages - there exists postquantum tripartite steering with two uncharacterised subsystems). To show this, consider assemblage defined as $\sigma^{(C)}_{ab|xy}=p^{(AB)}(ab|xy)\rho^{(C)}$ with no-signaling conditional probabilities (correlations) $p^{(AB)}(ab|xy)$ coming from PR-box (or any other non-local extremal point in the polytope of no-signaling boxes). Indeed, if there were a tripartite state $\rho^{ABC}$ and local measurements  given by elements of local POVMs $M^{(A)}_{a|x}$ and $N^{(B)}_{b|y}$ (acting on subsystem A and B respectively), such that $p^{(AB)}(ab|xy)\rho^{(C)}=\mathrm{Tr}_{AB}(M^{(A)}_{a|x}\otimes N^{(B)}_{b|y}\otimes \mathds{1}\rho^{(ABC)})$, then it would be true that $p^{(AB)}(ab|xy)=\mathrm{Tr}_{AB}(M^{(A)}_{a|x}\otimes N^{(B)}_{b|y}\rho^{(AB)})$ with $\rho^{(AB)}=\mathrm{Tr}_C(\rho^{(ABC)})$. The last statement leads to contradiction, as we know that non-local (nontrivial) extremal points in the polytope of no-signaling correlations (such as PR-box) do not admit quantum realisation \cite{our}. As we see postquantum nature of such assemblage is based purely on postquantum non-locality (the fact that not all no-signaling correlations can be reconstructed by quantum measurements \cite{our}). 

However, as it was shown (using SDP) in \cite{SBCSV15}, there are assemblages $\Sigma^{(C)}=\left\{\sigma^{(C)}_{ab|xy}\right\}_{a,b,x,y}$ without quantum realisation, for which with any choice of POVMs elements $R^{(C)}_{c|z}$, constructed  no-signaling boxes $p^{(ABC)}(abc|xyz)=\mathrm{Tr}(R^{(C)}_{c|z}\sigma^{(C)}_{ab|xy})$ admit a quantum realisation (in the usual sense for multipartite correlations). While non-locality can led to postquantum steering, in general it may not be the case. This observation enable us to treat postquantum steering (existence of no-signaling assemblages without quantum realisation) and postquantum non-locality as different phenomena. 
It is however interesting that, despite the fact that some no-signaling assemblages are not quantum, any no-signaling assemblage admits realisation given by (\ref{sup_q2}), if a state $\rho^{(ABC)}$ is in general exchange for tripartite hermitian operator $W^{(ABC)}$ with $\mathrm{Tr}(W^{(ABC)})=1$ \cite{SAPHS18}. Therefore, any abstract no-signaling assemblage has in a sense a concrete realisation.

In analogy with the bipartite steering, it is clear that the set of no-signaling assemblages and the subsets of quantum assemblages (assemblages with quantum realisation) are convex. The same is true for the subset of LHS assemblages - a no-signaling assemblage from tripartite steering admits LHS model \cite{SAPHS18} if it can be represented by 
\begin{equation}\label{sup_lhs2}
\sigma^{(C)}_{ab|xy}=\sum_i q_i p^{(A)}_i(a|x)p^{(B)}_i(b|y)\sigma^{(C)}_i
\end{equation}where $q_i\geq 0, \sum_i q_i=1$, $\sigma^{(C)}_i$ are some states of characterised subsystem C and $p^{(A)}_i(a|x),p^{(B)}_i(b|y)$ denotes conditional probability distributions for uncharacterised subsystem A and B respectively. One can see this LHS assemblages precisely as assemblages with quantum realisation obtained by measurements on tripartite fully separable states. Because conditional probabilities may be expressed as convex combinations of deterministic conditional probabilities, we may write $p^{(A)}_i(a|x)=\sum_k r_{ki} \tilde{p}^{(A)}_k(a|x)$ and $p^{(B)}_i(b|y)=\sum_l s_{li} \tilde{p}^{(B)}_l(b|y)$. If so, then for any LHS assemblage (\ref{sup_lhs2})
\begin{align}\nonumber
\sigma_{ab|xy}&=\sum_{i,k,l} q_ir_{ki}s_{li}\tilde{p}^{(A)}_k(a|x)\tilde{p}^{(B)}_l(b|y)\sigma^{(C)}_i\\ \nonumber
&=\sum_{k,l}\tilde{p}^{(A)}_k(a|x)\tilde{p}^{(B)}_l(b|y)\sum_i q_ir_{ki}s_{li}\sigma^{(C)}_i\\ \nonumber
&=\sum_{k,l}\tilde{p}^{(A)}_k(a|x)\tilde{p}^{(B)}_l(b|y) \tilde{q}_{kl}\tilde{\sigma}^{(C)}_{kl}\\ \label{sup_new_lhs}
&=\sum_j \tilde{q}_jp^{(AB)}_j(ab|xy)\tilde{\sigma}^{(C)}_j 
\end{align}with states $\tilde{\sigma}^{(C)}_{kl}=\frac{1}{\tilde{q}_{kl}}\sum_i q_ir_{ki}s_{li}\sigma^{(C)}_i$, $\tilde{q}_{kl}=\mathrm{Tr}(\sum_i q_ir_{ki}s_{li}\sigma^{(C)}_i)$ as coefficients of convex combination, index $j=(k,l)$ and $p^{(AB)}_j(ab|xy)$ defining some deterministic box of conditional probabilities. This provides a convenient presentation of LHS assemblages.

In the tripartite case distinction between fully separable states and biseparable states implies that one can consider additional class of biseparable assemblages, forming an intermediate set between LHS and quantum assemblages. An assemblage $\Sigma^{(C)}=\left\{\sigma^{(C)}_{ab|xy}\right\}$ is \textit{biseparable} \cite{CS15} if it can be expressed by a convex combination of the following form
\begin{align}\label{bisep}
\sigma_{ab|xy}=&\sum_i p_i p^{(A)}_i(a|x)\sigma^{(C),i}_{b|y}+\sum_j q_j p^{(B)}_j(b|y)\sigma^{(C),j}_{a|x}\nonumber \\
&+\sum_k r_k  p^{(AB)}_k(ab|xy)\rho^{(C)}_k
\end{align}where $P_{A,i}=\left\{p^{(A)}_i(a|x)\right\}$ are some local box on subsystem A, $P_{B,j}=\left\{p^{(B)}_j(b|y)\right\}$ are some local box on subsystem B, $P_{AB,k}=\left\{p^{(AB)}_k(ab|xy)\right\}$ are some quantum box on subsystem AB, $\Sigma_{i}^{(C)}=\left\{\sigma^{(C),i}_{b|y}\right\}$ are some bipartite no-signaling assemblage, $\Sigma^{(C)}_{j}=\left\{\sigma^{(C),j}_{a|x}\right\}$ are some bipartite no-signaling assemblage, $\rho^{(C)}_k$ are some states on subsystem C and $p_i,q_j,r_k\geq 0$ for all $i,j,k$ with $\sum_ip_i+\sum_jq_j+\sum_kr_k=1$.

One can easily see that quantum assemblage obtained by local measurements performed on two subsystems of tripartite biseparable state is in fact biseparble assemblage. Indeed, this is true as any such biseparable state $\rho^{(ABC)}$ is given as 
\begin{align}
\rho^{(ABC)}=&\sum_i p_i\rho^{(A)}_i\otimes \rho^{(BC)}_i+\sum_j q_j \rho^{(B)}_j\otimes \rho^{(AC)}_{j}\nonumber \\
&+\sum_k r_k \rho^{(AB)}_{k}\otimes \rho^{(C)}_k.
\end{align}In particular if quantum assemblage $\Sigma^{(C)}=\left\{\sigma_{ab|xy}^{(C)}\right\}$ define by formula $\sigma^{(C)}_{ab|xy}=\mathrm{Tr}_{AB}(M^{(A)}_{a|x}\otimes N^{(B)}_{b|y}\otimes \mathds{1}\rho^{(ABC)})$ is not biseparable, then the initial state $\rho^{(ABC)}$ is genuinely entangled (i.e. it s not biseparable). Moreover, it can be shown that any biseparable assemblage admits quantum realisation where local measurements are performed on some biseparable state (but some biseparable assemblages may admit quantum realisation with genuine entangled tripartite states).

Observe that the set of biseparable assemblages can be also defined as a set of all convex combination of assemblages of the form $\Sigma^{(C)}_{I}=\left\{\sigma_{ab|xy}^{(C),I}\right\}$, $\Sigma^{(C)}_{II}=\left\{\sigma_{ab|xy}^{(C),II}\right\}$ or $\Sigma^{(C)}_{III}=\left\{\sigma_{ab|xy}^{(C),III}\right\}$ where
\begin{equation}\label{sig_1}
\sigma_{ab|xy}^{(C),I}=p^{(A)}(a|x)\sigma^{(C),I}_{b|y}
\end{equation}for some local deterministic box $L_{I}=\left\{p^{(A)}(a|x)\right\}$ on subsystem A and some bipartite assemblage $\tilde{\Sigma}^{(C)}_{I}=\left\{\sigma_{b|y}^{(C),I}\right\}$,
\begin{equation}\label{sig_2}
\sigma_{ab|xy}^{(C),II}=p(b|y)\sigma^{(C),II}_{a|x}
\end{equation}for some local deterministic box $L_{II}=\left\{p(b|y)\right\}$ on subsystem B and some bipartite assemblage $\tilde{\Sigma}^{(C)}_{II}=\left\{\sigma_{a|x}^{(C),II}\right\}$,
\begin{equation}\label{sig_3}
\Sigma^{(C)}_{III}=P\otimes |\psi\rangle\langle\psi|
\end{equation}for some quantum box $P$ on subsystem AB and some pure state $|\psi\rangle$ on subsystem C.

Therefore, any optimization of linear function over set of biseparable assemblages can be reduced to optimization over assemblages given by one of the above forms (\ref{sig_1}-\ref{sig_3}). Moreover, extremal no-signaling assemblage belongs to the set of biseparable assemblages only if is one of the forms (\ref{sig_1}-\ref{sig_3}) - in particular this shows that any inflexible assemblage constructed according to subsection \ref{sec} cannot be biseparable.
\color{black}

In the main text of this paper we restrict our attention only to the simplest nontrivial case of tripartite steering with two uncharacterised subsystems. Namely we assume $a,b,x,y\in \left\{0,1\right\}$. In this case it is useful to rearrange all elements of a given assemblage into the following box
\begin{equation}\label{sup_box}
\Sigma^{(C)}=\begin{pmatrix}
\begin{array}{cc|cc}
 \sigma^{(C)}_{00|00} &  \sigma^{(C)}_{01|00} & \sigma^{(C)}_{00|01} &  \sigma^{(C)}_{01|01} \\  
 \sigma^{(C)}_{10|00} & \sigma^{(C)}_{11|00}& \sigma^{(C)}_{10|01}  & \sigma^{(C)}_{11|01} \\ \hline
 \sigma^{(C)}_{00|10} & \sigma^{(C)}_{01|10} & \sigma^{(C)}_{00|11}&  \sigma^{(C)}_{01|11}  \\
   \sigma^{(C)}_{10|10} & \sigma^{(C)}_{11|10} & \sigma^{(C)}_{10|11} & \sigma^{(C)}_{11|11}
\end{array}
\end{pmatrix}
\end{equation}where pairs $a,x$ label rows and pairs $b,y$ label columns. No-signaling conditions (\ref{sup_as2}-\ref{sup_as3}) have now simple interpretation. In each row sum of two operators on the right-hand side must be equal to the sum of two operators on the left-hand side. Similarly, in each column the sum of two operators in the upper part must be equal to the sum of two operators in the lower part. Normalisation (\ref{sup_as4}) is on the other hand encoded in the sum of traces of operators with the same pair $x,y$.

From the definition of LHS assemblage (see (\ref{sup_new_lhs})) it follows, that any such assemblage can be characterised by a convex combination of boxes (\ref{sup_box}) which have only four nonzero positions occupied by the same pure state forming a rectangle with exactly one element corresponding to each pair $x,y$, for example
\begin{equation}
\Sigma^{(C)}=\begin{pmatrix}
\begin{array}{cc|cc}
 |\phi\rangle \langle \phi| &  0 & 0 & |\phi\rangle \langle \phi|  \\  
 0 & 0 & 0  & 0 \\ \hline
 0 & 0 & 0 &  0  \\
 |\phi\rangle \langle \phi| & 0 & 0 & |\phi\rangle \langle \phi|
\end{array}
\end{pmatrix}.
\end{equation}Note that each assemblage of this form is inflexible (see Section \ref{sec}), hence it is an extremal point in the set of all no-signaling assemblages (obviously it is then an extremal point in the set of LHS assemblages as well). In particular, all problems based on optimization over set of LHS assemblages my be reduced to optimization over set of this extremal points. Because of the previous discussion, this optimization boils down to optimization over pure states $|\phi\rangle$ (from $d_C$ dimensional space describing assemblage) and deterministic boxes of conditional probabilities (which describe which four position in the box (\ref{sup_box}) are occupied by $|\phi\rangle \langle \phi|$). Note that in considering setting (i.e. $a,b,x,y\in \left\{0,1\right\}$) there are 16 deterministic boxes $L_{\alpha\beta\gamma\delta}=\left\{p^{(AB)}_{\alpha\beta\gamma\delta}(ab|xy)\right\}_{a,b,x,y}$, defined by conditional probabilities ($\oplus$ stands here for addition modulo 2 and $\alpha,\beta, \gamma, \delta\in\left\{0,1\right\}$).
\begin{equation}
 p^{(AB)}_{\alpha\beta\gamma\delta}(ab|xy)=
\begin{cases}
1\ \ \ \mathrm{for}\ a=\alpha x\oplus \beta, b=\gamma y\oplus \delta \\
0 \ \ \  \mathrm{otherwise}.
\end{cases}
\end{equation}This explains remark after after Proposition 4.

As it was shown not all no-signaling assemblages admit quantum realisation. In particular one can find extremal points in the set of no-signaling assemblages without quantum realisation. Indeed, consider assemblage $\sigma^{(C)}_{ab|xy}=p^{(AB)}(ab|xy)|\phi^{(C)}\rangle \langle \phi^{(C)}|$ with $p^{(AB)}(ab|xy)$ coming from non-local (nontrivial) extremal points in the polytope of no-signaling correlations (no-signaling boxes). This assemblage is inflexible (see Section \ref{sec}) and therefore extremal. In the next subsections we introduce a set of tools which together with reasoning presented in the main text provide a construction of extremal no-signaling assemblages with quantum realisation. Note that existence of such an extremal points provide a contrast with the theory of polytopes of no-signaling correlations (no-signaling boxes) in which similar situation is impossible, regardless of the number of parties, settings or outcomes \cite{our}.

Note that all presented notions can be further extended to the general multipartie setting with $n$ uncharacterised parties performing a measurements on a joint $n+1$ party system (see \cite{SAPHS18}). In this multipartite steering scenario each party related to uncharacterised subsystem $A_i$ perform local measurements labeled by $x_i\in \mathcal{X}_i$ obtaining outcomes labeled by $a_i\in \mathcal{A}_i$. Probabilistic description of a state of characterised quantum subsystem $C$ conditioned upon indexes of measurements $\mathbf{x}_n=(x_1,\ldots,x_n)$ and outcomes $\mathbf{a}_n=(a_1,\ldots,a_n)$ in this experimental setting under no-signaling (but non necessarily quantum) constrains is once more given by the notion of no-signaling assemblage. Multipartite no-signaling assemblage $\Sigma^{(C)}=\left\{\sigma^{(C)}_{\mathbf{a}_n|\mathbf{x}_n}\right\}_{\mathbf{a}_n,\mathbf{x}_n}$ is therefore a collection of subnormalised states $\sigma_{\mathbf{a}^{(C)}_n|\mathbf{x}_n}$ (acting on $d_C$-dimensional Hilbert space describing characterised subsystem $C$) such that
\begin{equation}\label{def11}
\forall_{\mathbf{x}_n} \sum_{\mathbf{a}_n} \sigma^{(C)}_{\mathbf{a}_n|\mathbf{x}_n}=\sigma_C,
\end{equation}where $\sigma_C$ is some state, and for any set of indexes $I=\left\{i_1,\ldots, i_s\right\}\subset \left\{1,\ldots, n\right\}$ with $1\leq s<n$ there exist an operator $\sigma^{(C)}_{a_{i_1}\ldots a_{i_s}|x_{i_1}\ldots x_{i_s}}$ that fulfills 
\begin{equation}\label{def12}
\forall_{a_k,k\in I}\forall_{\mathbf{x}_n} \sum_{a_j,j\notin I} \sigma^{(C)}_{\mathbf{a}_n|\mathbf{x}_n}=\sigma^{(C)}_{a_{i_1}\ldots a_{i_s}|x_{i_1}\ldots x_{i_s}}.
\end{equation}
Assemblage $\Sigma^{(C)}=\left\{\sigma^{(C)}_{\mathbf{a}_n|\mathbf{x}_n}\right\}$ of this form is quantum (i.e. admits quantum realisation) if there exist some quantum state $\rho^{(A_1\ldots A_nC)}$ of some composed system $A_1\ldots A_nC$ and some local POVMs $M^{(A_1)}_{a_1|x_1}, \ldots, M^{(A_n)}_{a_n|x_n}$ such that
\begin{equation}
\sigma_{\mathbf{a}_n|\mathbf{x}_n}^{(C)}=\mathrm{Tr}_{A_1\ldots A_n}(M^{(A_1)}_{a_1|x_1}\otimes \ldots M^{(A_n)}_{a_n|x_n}\otimes \mathds{1}\rho^{(A_1\ldots A_nC)}.
\end{equation}Note that both set of no-signaling and quantum assemblages are convex sets. Once more, similar to the above discussion, one may also introduce convex subsets of quantum assemblages given by all assemblages admitting LHS model \cite{SAPHS18} of the form
\begin{equation}\label{sup_lhsxxx}
\sigma^{(C)}_{\mathbf{a}_n|\mathbf{x}_n}=\sum_i q_i p^{(A_1)}_i(a_1|x_1)\ldots p^{(A_n)}_i(a_n|x_n)\sigma^{(C)}_i
\end{equation}where $q_i\geq 0, \sum_i q_i=1$, $\sigma^{(C)}_i$ are some states of characterised subsystem C and $p^{(A_j)}_i(a_j|x_j)$ denotes conditional probability distributions for uncharacterised subsystem $A_j$ respectively. In a natural way one can also consider general multipartite biseparable assemblages (forming a convex set as well) defined as those assemblages with quantum realisation given by local measurements performed on some biseparable (i.e. not genuine entangled) state. 

\subsection{Notion on inflexibility in the multipartite case}\label{multipartite_case}

In order to tackle the problem of quantum realisation of extremal no-signaling assemblage we firstly introduce the notion of inflexibilty of assemblages of pure state in the full generality for multipartite steering scenario with uncharcterised subsystems $A_i,\ldots A_n$ and characterised subsystem $C$.

We say that a given no-signaling assemblage $\Sigma^{(C)}=\left\{\sigma^{(C)}_{\mathbf{a}_n|\mathbf{x}_n}\right\}_{\mathbf{a}_n,\mathbf{x}_n}$ is an assemblage of pure states if for any index $\mathbf{a}_n,\mathbf{x}_n$ rank of $\sigma^{(C)}_{\mathbf{a}_n|\mathbf{x}_n}$ is not greater than one (thus in general such assemblage may consist of subnormalized pure states and zeros). Note that given assemblage of pure states may not be extremal in the set of all no-signaling assemblages, as well as there are extremal no-signaling assembles that are not assemblages of pure states. 

Fixed any no-signaling assemblage of pure states $\Sigma^{(C)}=\left\{p_{\mathbf{a}_n|\mathbf{x}_n}|\psi_{\mathbf{a}_n|\mathbf{x}_n}\rangle \langle \psi_{\mathbf{a}_n|\mathbf{x}_n}|\right\}_{\mathbf{a}_n|\mathbf{x}_n}$. Consider assemblages of pure states $\hat{\Sigma}^{(C)}=\left\{q_{\mathbf{a}_n|\mathbf{x}_n}|\psi_{\mathbf{a}_n|\mathbf{x}_n}\rangle \langle \psi_{\mathbf{a}_n|\mathbf{x}_n}|\right\}_{a,b,x,y}$ where all pure states $|\psi_{\mathbf{a}_n|\mathbf{x}_n}\rangle \langle \psi_{\mathbf{a}_n|\mathbf{x}_n}|$ are the same as in $\Sigma^{(C)}$. Finally, let $S_{\Sigma^{(C)}}$ denote the set of all such $\hat{\Sigma}^{(C)}$ for which $p_{\mathbf{a}_n|\mathbf{x}_n}=0$ implies $q_{\mathbf{a}_n|\mathbf{x}_n}=0$ for any index $\mathbf{a}_n,\mathbf{x}_n$. We say that assemblage $\tilde{\Sigma}^{(C)}\in S_{\Sigma^{(C)}}$ is \textit{similar} to $\Sigma^{(C)}$ and we say that an assemblage of pure states $S_{\Sigma^{(C)}}$ is \textit{inflexible}, if the set $S_{\Sigma^{(C)}}$ consists of only one element (namely $\Sigma^{(C)}$ itself).

To justify introduction of this notion, consider a convex decomposition of some assemblage of pure states $\Sigma^{(C)}=p\Sigma^{(C)}_1+(1-p)\Sigma^{(C)}_2$, with $p\in(0,1)$ and $\Sigma^{(C)}_i=\left\{\sigma_{\mathbf{a}_n|\mathbf{x}_n}^{(C,i)}\right\}_{\mathbf{a}_n,\mathbf{x}_n}$. For a given index $\mathbf{a}_n|\mathbf{x}_n$ we have

\begin{equation}\label{ex_pure}
p_{\mathbf{a}_n|\mathbf{x}_n}|\psi_{\mathbf{a}_n|\mathbf{x}_n}\rangle\langle \psi_{\mathbf{a}_n|\mathbf{x}_n}|=p\sigma^{(1)}_{\mathbf{a}_n|\mathbf{x}_n}+(1-p)\sigma^{(2)}_{\mathbf{a}_n|\mathbf{x}_n}.
\end{equation}

Let $p_{\mathbf{a}_n|\mathbf{x}_n}=0$ and $\sigma^{(C,i)}_{\mathbf{a}_n|\mathbf{x}_n}\neq 0$ for $i=1,2$. Because $p_{\mathbf{a}|\mathbf{x}}=p\mathrm{Tr}(\sigma^{(1)}_{\mathbf{a}|\mathbf{x}})+(1-p)\mathrm{Tr}(\sigma^{(2)}_{\mathbf{a}_n|\mathbf{x}_n})$ we get the following convex combination of states
\begin{widetext}
\begin{equation}\label{ex_pure}
|\psi_{\mathbf{a}_n|\mathbf{x}_n}\rangle\langle \psi_{\mathbf{a}_n|\mathbf{x}_n}|=\frac{p\mathrm{Tr}(\sigma^{(C,1)}_{\mathbf{a}_n|\mathbf{x}_n})}{p_{\mathbf{a}_n|\mathbf{x}_n}}\frac{\sigma^{(C,1)}_{\mathbf{a}_n|\mathbf{x}_n}}{\mathrm{Tr}(\sigma^{(C,1)}_{\mathbf{a}_n|\mathbf{x}_n})}+\frac{(1-p)\mathrm{Tr}(\sigma^{(C,2)}_{\mathbf{a}_n|\mathbf{x}_n})}{p_{\mathbf{a}_n|\mathbf{x}_n}}\frac{\sigma^{(C,2)}_{\mathbf{a}_n|\mathbf{x}_n}}{\mathrm{Tr}(\sigma^{(C,2)}_{\mathbf{a}_n|\mathbf{x}_n})}.
\end{equation}
\end{widetext}From this $\sigma^{(C,i)}_{\mathbf{a}_n|\mathbf{x}_n}=\mathrm{Tr}(\sigma^{(C,i)}_{\mathbf{a}_n|\mathbf{x}_n})|\psi_{\mathbf{a}_n|\mathbf{x}_n}\rangle\langle \psi_{\mathbf{a}_n|\mathbf{x}_n}|$ for $i=1,2$. By considering remaining trivial possibilities ($p_{\mathbf{a}_n|\mathbf{x}_n}=0$ or without loss of generality $p_{\mathbf{a}_n|\mathbf{x}_n}\neq 0$ with $\sigma^{(C,1)}_{\mathbf{a}_n|\mathbf{x}_n}=0)$ we show that $\Sigma^{(C)}=p\Sigma^{(C)}_1+(1-p)\Sigma^{(C)}_2$ implies $\Sigma^{(C)}_1,\Sigma^{(C)}_2\in S_{\Sigma^{(C)}}$. Therefore, inflexibility implies extremality.

For a given inflexible assemblage $\Sigma^{(C)}=\left\{\sigma^{(C)}_{\mathbf{a}_n|\mathbf{x}_n}\right\}_{\mathbf{a}_n,\mathbf{x}_n}$ we may define operators 
\begin{equation}
 \rho^{(C)}_{\mathbf{a}_n|\mathbf{x}_n}=
\begin{cases}
0 & \mbox{for $\sigma^{(C)}_{\mathbf{a}_n|\mathbf{x}_n}=0$,} \\
\frac{\sigma^{(C)}_{\mathbf{a}_n|\mathbf{x}_n}}{\mathrm{Tr}(\sigma^{(C)}_{\mathbf{a}_n|\mathbf{x}_n})} & \mbox{for $\sigma^{(C)}_{\mathbf{a}_n|\mathbf{x}_n}\neq 0 $.}
\end{cases} 
\end{equation}and introduce a linear functional $F_{\Sigma^{(C)}}$ given by the formula
\begin{equation}\label{expression}
F_{\Sigma^{(C)}}(\tilde{\Sigma}^{(C)})=\sum_{\mathbf{a}_n,\mathbf{x}_n}\mathrm{Tr}(\rho^{(C)}_{\mathbf{a}_n|\mathbf{x}_n}\tilde{\sigma}^{(C)}_{\mathbf{a}_n|\mathbf{x}_n})
\end{equation}where $\tilde{\Sigma}^{(C)}$ is any no-signaling assemblage. Observe that 
\begin{equation}
F_{\Sigma^{(C)}}(\tilde{\Sigma}^{(C)})\leq \prod_{i=i}^n|\mathcal{X}_i|
\end{equation}for any $\tilde{\Sigma}^{(C)}$ and $\leq$ is replace by equality if and only if $\tilde{\Sigma}^{(C)}\in S_{\Sigma^{(C)}}$. Because of that inflexible assemblage $\Sigma^{(C)}$ is the unique no-signaling assemblage for which functional $F_{\Sigma^{(C)}}$ attains maximal value over the set of all no-signaling assemblages. In other words, not only inflexible assemblages are extremal, but they are exposed as well.

Finally, it is worth noting that if inflexible assemblage $\Sigma^{(C)}$ used in the above construction of $F_{\Sigma^{(C)}}$ does not admit LHS model or is not biseparable, then linear functional $F_{\Sigma^{(C)}}$ defines a steering inequality for multipartite scenario detecting entangled or genuinely entangled multipartite states respectively - see example of this in Theorem 4 in the main text).
 
\color{black}

\subsection{Inflexibility of a certain class of assemblages}\label{sec}

Here we introduce several concepts and use them to formulate and prove an important lemma concerning a certain class of assemblages.

Consider a steering scenario in which two uncharacterised parties steer a characterised one by measurements described respectively by labels of settings and outcomes $a,b,x,y\in \left\{0,1\right\}$. Let $\Sigma^{(C)}$ be a general no-signaling assemblage (possibly postquantum) related to that scenario and with all position occupied by at most rank one operators ($p_i\geq 0$), i.e. let $\Sigma^{(C)}$ be an assemblage of pure states
\begin{widetext}
\begin{equation}\label{ass_pure}
\Sigma^{(C)}=\begin{pmatrix}
\begin{array}{cc|cc}
   p_1|\psi_1\rangle \langle \psi_1| &  p_2|\psi_2\rangle \langle \psi_2| &  p_3|\psi_3\rangle \langle \psi_3| & p_4|\psi_4\rangle \langle \psi_4|\\[8pt]
    p_5|\psi_5\rangle \langle \psi_5| & p_6|\psi_6\rangle \langle \psi_6|  & p_7|\psi_7\rangle \langle \psi_7| &  p_8|\psi_8\rangle \langle \psi_8| \\[8pt] \hline
		 p_9|\psi_9\rangle \langle \psi_9| &  p_{10}|\psi_{10}\rangle \langle \psi_{10}| &  p_{11}|\psi_{11}\rangle \langle \psi_{11}| & p_{12}|\psi_{12}\rangle \langle \psi_{12}|\\[8pt]
    p_{13}|\psi_{13}\rangle \langle \psi_{13}| & p_{14}|\psi_{14}\rangle \langle \psi_{14}|  & p_{15}|\psi_{15}\rangle \langle \psi_{15}| &  p_{16}|\psi_{16}\rangle \langle \psi_{16}| 
\end{array}
\end{pmatrix}.
\end{equation}
\end{widetext}Now take any row or any column from this assemblage - without loss of generality let us focus on a first row
\begin{equation}\label{row}
\begin{matrix}
\begin{array}{cc|cc}
   p_1|\psi_1\rangle \langle \psi_1|& p_2|\psi_2\rangle \langle \psi_2|&  p_3|\psi_3\rangle \langle \psi_3| & p_4|\psi_4\rangle \langle \psi_4|.
\end{array}
\end{matrix}
\end{equation}As a part of no-signaling assemblage, this row must satisfy $p_1|\psi_1\rangle \langle \psi_1|+p_2|\psi_2\rangle \langle \psi_2|=p_3|\psi_3\rangle \langle \psi_3|+p_4|\psi_4\rangle \langle \psi_4|$. This is possible only in one of three mutually exclusive instances. Firstly, all $|\psi_i\rangle \langle \psi_i|$ may be the same - then we will say that row (column) is \textit{type I}. Secondly, if $|\psi_1\rangle \langle \psi_1|$ and $|\psi_2\rangle \langle \psi_2|$ are different (i.e. $|\psi_1\rangle$ and $|\psi_2\rangle$ are linearly independent) and remaining pair of $|\psi_i\rangle \langle \psi_i|$ is the same as the first one (up to relabeling), then we will say that row (column) is \textit{type II}. Take any such row (column) without loss of generality given as
\begin{equation}\nonumber
\begin{matrix}
\begin{array}{cc|cc}
   p_1|\psi_1\rangle \langle \psi_1|& p_2|\psi_2\rangle \langle \psi_2|&  p_3|\psi_1\rangle \langle \psi_1| & p_4|\psi_2\rangle \langle \psi_2|.
\end{array}
\end{matrix}
\end{equation}Assume that $p_1=0$, then $p_3=0$ as well (since for type II coefficients $p_i$ related to the same projection must be equal, i.e. $p_1=p_3,p_2=p_4$) and the considered row (column) is described by the single rank one operator $|\psi_2\rangle \langle \psi_2|$ - therefore it is indistinguishable from row (column) of type I. Because of that observation, from now on we will always assume that in case of type II all $p_i$ are nonzero. Finally, it may happen that all $|\psi_i\rangle \langle \psi_i|$ are different - then we will say that row (column) is \textit{type III}. Note that in the case of type III all $p_i$ are nonzero. Moreover, if a given row (column) of the form (\ref{row}) is type III, then any other valid row (column) of no-signaling assemblage with the same $|\psi_i\rangle \langle \psi_i|$ must be of the following form
\begin{equation}\label{row2}
\begin{matrix}
\begin{array}{cc|cc}
    q_1|\psi_1\rangle \langle \psi_1|&  q_2|\psi_2\rangle \langle \psi_2|&  q_3|\psi_3\rangle \langle \psi_3| & q_4|\psi_4\rangle \langle \psi_4|
\end{array}
\end{matrix}
\end{equation}where $q_i=\alpha p_i$ and $\alpha$ is some non-negative constant. Indeed, without loss of generality we can take positive $\alpha$ such that $\alpha p_1=q_1$. Then if we multiply by $\alpha$ the first row (the one with $p_i$) and subtract it from the second (the one with $q_i$) we will have (due to no-signaling constraints) that $(q_2 - \alpha p_2)|\psi_2\rangle \langle \psi_2| = (q_3 - \alpha p_3)|\psi_3 \rangle \langle \psi_3| + (q_4-\alpha p_4)|\psi_4 \rangle \langle \psi_4|$ and this is only true for $q_i=\alpha p_i$ since all $|\psi_i \rangle \langle \psi_i|$ are different.

Recall that if $\Sigma^{(C)}=\left\{p_i|\psi^{(C)}_i\rangle \langle \psi^{(C)}_i|\right\}_i$ and $\tilde{\Sigma}^{(C)}=\left\{q_i|\psi^{(C)}_i\rangle \langle \psi^{(C)}_i|\right\}_i$, then $\tilde{\Sigma}^{(C)}$ is \textit{similar} to $\Sigma^{(C)}$ when $p_i=0$ implies $q_i=0$ (see Definition 2 in the main text and Subsection \ref{multipartite_case}). Note that the notion of similarity is not a symmetric relation. It may happen that $\tilde{\Sigma}^{(C)}$ is similar to $\Sigma^{(C)}$, but $\Sigma^{(C)}$ is not similar to $\tilde{\Sigma}^{(C)}$ (if $\tilde{\Sigma}^{(C)}$ has more positions occupied by $0$ than $\Sigma^{(C)}$). 

As we show in Subsection \ref{multipartite_case} \textit{inflexibility implies extremality}, since convex decomposition of assebmlage f pure states $\Sigma^{(C)}=p\Sigma^{(C)}_1+(1-p)\Sigma^{(C)}_2$ is possible only if both $\Sigma^{(C)}_i$ are similar to $\Sigma^{(C)}$.

Now we are in the position to exploit the concepts introduced above. As we shall see below careful application of no-signaling constraints leads to the following lemma, providing a sufficient condition for inflexibility (hence extremality and exposedness), crucial for the proof of the Proposition 4 in the main text.

\begin{lem}\label{lemma1}
Let $\Sigma^{(C)}$ be a pure no-signaling assemblage with a single column of type III and two rows (forming a pair in upper or lower part of assemblage) of type III. Then $\Sigma^{(C)}$ is inflexible.
\end{lem}
\begin{proof}We may consider assemblage where first column and first and second rows are type III. Let us discuss all possibilities (without loss of generality).

Firstly, assume that the third row of considered assemblage (\ref{ass_pure}) is type III. Now because of previous discussion for rows of type III, any assemblage $\tilde{\Sigma}^{(C)}$ similar to $\Sigma^{(C)}$ must be of the form 
\begin{widetext}
\begin{equation}
\tilde{\Sigma}^{(C)}=\begin{pmatrix}
\begin{array}{cc|cc}
  \alpha_1 p_1|\psi_1\rangle \langle \psi_1| & \alpha_1 p_2|\psi_2\rangle \langle \psi_2| & \alpha_1 p_3|\psi_3\rangle \langle \psi_3| & \alpha_1 p_4|\psi_4\rangle \langle \psi_4|\\[8pt]
    \alpha_2 p_5|\psi_5\rangle \langle \psi_5| & \alpha_2 p_6|\psi_6\rangle \langle \psi_6|  & \alpha_2 p_7|\psi_7\rangle \langle \psi_7| & \alpha_2  p_8|\psi_8\rangle \langle \psi_8| \\[8pt]\hline
		\alpha_3 p_9|\psi_9\rangle \langle \psi_9| & \alpha_3 p_{10}|\psi_{10}\rangle \langle \psi_{10}| & \alpha_3 p_{11}|\psi_{11}\rangle \langle \psi_{11}| &\alpha_3 p_{12}|\psi_{12}\rangle \langle \psi_{12}|\\[8pt]
  \beta_1  p_{13}|\psi_{13}\rangle \langle \psi_{13}| &\beta_2 p_{14}|\psi_{14}\rangle \langle \psi_{14}|  &\beta_3 p_{15}|\psi_{15}\rangle \langle \psi_{15}| &\beta_4  p_{16}|\psi_{16}\rangle \langle \psi_{16}| 
\end{array}
\end{pmatrix}.
\end{equation}
\end{widetext}But since the first column of $\Sigma^{(C)}$ is type III as well we have $\alpha=\alpha_1=\alpha_2=\alpha_3=\beta_1$. Moreover, in order to satisfy no-signaling conditions in columns we have to put $\alpha=\beta_i$ for all $i$. Then, by normalisation constraint, $\alpha=1$ and $\Sigma^{(C)}$ is indeed inflexible in this case.\newline

Now assume that the third and the fourth rows of considered assemblage $\Sigma^{(C)}$ are type I. Any assemblage $\tilde{\Sigma}^{(C)}$ similar to $\Sigma$ must be of the form
\begin{widetext}
\begin{equation}\label{mat_gen}
\tilde{\Sigma}^{(C)}=\begin{pmatrix}
\begin{array}{cc|cc}
  \alpha p_1|\psi_1\rangle \langle \psi_1| & \alpha p_2|\psi_2\rangle \langle \psi_2| & \alpha p_3|\psi_3\rangle \langle \psi_3| & \alpha p_4|\psi_4\rangle \langle \psi_4|\\[8pt]
    \alpha p_5|\psi_5\rangle \langle \psi_5| & \alpha p_6|\psi_6\rangle \langle \psi_6|  & \alpha p_7|\psi_7\rangle \langle \psi_7| & \alpha  p_8|\psi_8\rangle \langle \psi_8| \\[8pt]\hline
		\alpha p_9|\psi_9\rangle \langle \psi_9| & \beta_2 p_{10}|\psi_{9}\rangle \langle \psi_{9}| & \beta_3 p_{11}|\psi_{9}\rangle \langle \psi_{9}| &\beta_4 p_{12}|\psi_{9}\rangle \langle \psi_{9}|\\[8pt]
  \alpha  p_{13}|\psi_{10}\rangle \langle \psi_{10}| &\gamma_2 p_{14}|\psi_{10}\rangle \langle \psi_{10}|  &\gamma_3 p_{15}|\psi_{10}\rangle \langle \psi_{10}| &\gamma_4  p_{16}|\psi_{10}\rangle \langle \psi_{10}| 
\end{array}
\end{pmatrix}.
\end{equation}
\end{widetext}Since first column of $\Sigma^{(C)}$ is type III, $|\psi_{10}\rangle \langle \psi_{10}|$ is different from $|\psi_{9}\rangle \langle \psi_{9}|$. We will show that this implies $\alpha=\beta_i=\gamma_i$ for all $i$. Indeed, this is true if columns are type II or III. If one of columns is not, for example we have $|\psi_2\rangle \langle \psi_2|=|\psi_6\rangle \langle \psi_6|$, then either $p_{10}=0$ or $p_{14}=0$. Assume that $p_{10}=0$ (the other case is analogous). We may conclude that $\alpha=\gamma_2$ and $\beta_2$ may be arbitrary. However, observe that since $p_{10}=0$, the form of \ref{mat_gen} does not depend on the value of $\beta_2$. If so, then without loss of generality we can always put $\alpha=\gamma_2=\beta_2$ (and repeat this reasoning for any $i$ if necessary). Finally, by normalisation constraint we see that $\Sigma^{(C)}$ must be inflexible.\newline

Consider the case when the third row of assemblage $\Sigma^{(C)}$ is type I and the fourth row is type II. Any assemblage $\tilde{\Sigma}^{(C)}$ similar to $\Sigma^{(C)}$ must be of the form (up to relabeling in the last row)
\begin{widetext}
\begin{equation}
\tilde{\Sigma}^{(C)}=\begin{pmatrix}
\begin{array}{cc|cc}
  \alpha p_1|\psi_1\rangle \langle \psi_1| & \alpha p_2|\psi_2\rangle \langle \psi_2| & \alpha p_3|\psi_3\rangle \langle \psi_3| & \alpha p_4|\psi_4\rangle \langle \psi_4|\\[8pt]
    \alpha p_5|\psi_5\rangle \langle \psi_5| & \alpha p_6|\psi_6\rangle \langle \psi_6|  & \alpha p_7|\psi_7\rangle \langle \psi_7| & \alpha  p_8|\psi_8\rangle \langle \psi_8| \\[8pt]\hline
		\alpha p_9|\psi_9\rangle \langle \psi_9| & \beta_2 p_{10}|\psi_{9}\rangle \langle \psi_{9}| & \beta_3 p_{11}|\psi_{9}\rangle \langle \psi_{9}| &\beta_4 p_{12}|\psi_{9}\rangle \langle \psi_{9}|\\[8pt]
  \alpha  p_{13}|\psi_{10}\rangle \langle \psi_{10}| &\gamma_2 p_{14}|\tilde{\psi}_{10}\rangle \langle \tilde{\psi}_{10}|  &\gamma_3 p_{13}|\psi_{10}\rangle \langle \psi_{10}| &\gamma_4  p_{14}|\tilde{\psi}_{10}\rangle \langle \tilde{\psi}_{10}| 
\end{array}
\end{pmatrix}.
\end{equation}
\end{widetext}Because first column of $\Sigma^{(C)}$ is type III, $|\psi_{10}\rangle \langle \psi_{10}|$ is different from $|\psi_{9}\rangle \langle \psi_{9}|$, so by previous case we obtain $\alpha=\beta_3=\gamma_3$. Note that the first row and the second row are type III, so both $p_4$ and $p_8$ are non-zero (the same is true for $p_2$ and $p_6$). Assume that $|\psi_{9}\rangle \langle \psi_{9}|=|\tilde{\psi}_{10}\rangle \langle \tilde{\psi}_{10}|$. Analysis of the fourth column shows that $|\psi_{4}\rangle \langle \psi_{4}|=|\psi_{8}\rangle \langle \psi_{8}|=|\psi_{9}\rangle \langle \psi_{9}|$. The same analysis of the second column shows that $|\psi_{2}\rangle \langle \psi_{2}|=|\psi_{6}\rangle \langle \psi_{6}|=|\psi_{9}\rangle \langle \psi_{9}|$. Therefore, $|\psi_{2}\rangle \langle \psi_{2}|=|\psi_{4}\rangle \langle \psi_{4}|$ which is in contradiction with the fact that first row of $\Sigma^{(C)}$ is type III. Therefore, $|\tilde{\psi}_{10}\rangle \langle \tilde{\psi}_{10}|\neq|\psi_{9}\rangle \langle \psi_{9}|$ and eventually we have $\alpha=\beta_i=\gamma_i$ for all $i$ which shows that $\Sigma^{(C)}$ is inflexible.\newline

\color{black}To conclude the proof assume the last case. Let the third and the fourth rows of considered assemblage $\Sigma^{(C)}$ be type II. Then depending on particular form of $\Sigma^{(C)}$, any similar assemblage $\tilde{\Sigma}^{(C)}$ is given (up to relabeling) either by 
\begin{widetext}
\begin{equation}\label{last_1}
\tilde{\Sigma}^{(C)}=\begin{pmatrix}
\begin{array}{cc|cc}
  \alpha p_1|\psi_1\rangle \langle \psi_1| & \alpha p_2|\psi_2\rangle \langle \psi_2| & \alpha p_3|\psi_3\rangle \langle \psi_3| & \alpha p_4|\psi_4\rangle \langle \psi_4|\\[8pt]
    \alpha p_5|\psi_5\rangle \langle \psi_5| & \alpha p_6|\psi_6\rangle \langle \psi_6|  & \alpha p_7|\psi_7\rangle \langle \psi_7| & \alpha  p_8|\psi_8\rangle \langle \psi_8| \\[8pt]\hline
		\alpha p_9|\psi_9\rangle \langle \psi_9| & \beta_2 p_{10}|\tilde{\psi}_{9}\rangle \langle \tilde{\psi}_{9}| & \beta_3 p_{9}|\psi_{9}\rangle \langle \psi_{9}| &\beta_4 p_{10}|\tilde{\psi}_{9}\rangle \langle \tilde{\psi}_{9}|\\[8pt]
  \alpha  p_{13}|\psi_{10}\rangle \langle \psi_{10}| &\gamma_2 p_{14}|\tilde{\psi}_{10}\rangle \langle \tilde{\psi}_{10}|  &\gamma_3 p_{13}|\psi_{10}\rangle \langle \psi_{10}| &\gamma_4  p_{14}|\tilde{\psi}_{10}\rangle \langle \tilde{\psi}_{10}|
\end{array}
\end{pmatrix}
\end{equation}or by
\begin{equation}\label{last_2}
\tilde{\Sigma}^{(C)}=\begin{pmatrix}
\begin{array}{cc|cc}
  \alpha p_1|\psi_1\rangle \langle \psi_1| & \alpha p_2|\psi_2\rangle \langle \psi_2| & \alpha p_3|\psi_3\rangle \langle \psi_3| & \alpha p_4|\psi_4\rangle \langle \psi_4|\\[8pt]
    \alpha p_5|\psi_5\rangle \langle \psi_5| & \alpha p_6|\psi_6\rangle \langle \psi_6|  & \alpha p_7|\psi_7\rangle \langle \psi_7| & \alpha  p_8|\psi_8\rangle \langle \psi_8| \\[8pt]\hline
		\alpha p_9|\psi_9\rangle \langle \psi_9| & \beta_2 p_{10}|\tilde{\psi}_{9}\rangle \langle \tilde{\psi}_{9}| & \beta_3 p_{10}|\tilde{\psi}_{9}\rangle \langle \tilde{\psi}_{9}| &\beta_4 p_{9}|\psi_{9}\rangle \langle \psi_{9}|\\[8pt]
  \alpha  p_{13}|\psi_{10}\rangle \langle \psi_{10}| &\gamma_2 p_{14}|\tilde{\psi}_{10}\rangle \langle \tilde{\psi}_{10}|  &\gamma_3 p_{13}|\psi_{10}\rangle \langle \psi_{10}| &\gamma_4  p_{14}|\tilde{\psi}_{10}\rangle \langle \tilde{\psi}_{10}|
\end{array}
\end{pmatrix}.
\end{equation}
\end{widetext}Consider first possibility (\ref{last_1}). This case can be treated in the same way as the case with one row of type I and one row of type II. Assume then that $\tilde{\Sigma}^{(C)}$ is of type (\ref{last_2}). Because first and second rows are type III, there may be at most two columns with two equal states in lower part, i.e. $|\psi_{10}\rangle \langle \psi_{10}|=|\tilde{\psi}_{9}\rangle \langle \tilde{\psi}_{9}|$ and $|\tilde{\psi}_{10}\rangle \langle \tilde{\psi}_{10}|=|\psi_{9}\rangle \langle \psi_{9}|$ - this is the only nontrivial situation to consider. If this is the case, obviously $\alpha=\beta_2=\gamma_2$ and since third and fourth rows are type II, we also have $\alpha=\beta_3=\beta_4$ and $\alpha=\gamma_3=\gamma_4$. Therefore, due to normalisation the proof is completed.
\end{proof}

\subsection{Extremal quantum assemblages from pure states}\label{simp}

Consider a quantum assemblage given by formula 
\begin{equation}
\sigma^{(C)}_{ab|xy}=\mathrm{Tr}_{AB}(M^{(A)}_{a|x}\otimes N^{(B)}_{b|y}\otimes \mathds{1}\rho^{(ABC)})
\end{equation}and an eigendecomposition 
\begin{equation}\nonumber
\rho^{(ABC)}=\sum_i \lambda_i|\psi^{(ABC)}_{i}\rangle \langle \psi^{(ABC)}_{i}|.
\end{equation}Obviously, $\sigma^{(C)}_{ab|xy}=\sum_i \lambda_i \sigma^{(C),i}_{ab|xy}$ where for each $i$
\begin{equation}\nonumber
\sigma_{ab|xy}^{(C),i}=\mathrm{Tr}_{AB}(M^{(A)}_{a|x}\otimes N^{(B)}_{b|y}\otimes \mathds{1}|\psi^{(ABC)}_{i}\rangle \langle \psi^{(ABC)}_{i}|)
\end{equation}defines another quantum assemblage. If original assemblage is extremal then $\sigma^{(C),i}_{ab|xy}=\sigma^{(C)}_{ab|xy}$ for any $i$. Therefore, any quantum assemblage which is also extremal in the set of all no-signaling assemblages can be obtained by measurements performed on a pure state.

\subsection{Proof of Proposition 4 - auxiliary Lemmas}

In order to conclude the proof of Proposition 4 we need to justify existence of appropriate PVMs, which performed on a considered genuine entangled state, give a rise to the quantum assemblage with first column and first and second row of type III (see sketch of proof of Proposition 4 and Subsection \ref{sec}). This follows directly from the following simple auxiliary Lemmas which we state here for the sake of completeness. 

\begin{lem}\label{lemmaA} 
Let  $|\psi^{(ABE)}\rangle\in \mathbb{C}^{2}\otimes \mathbb{C}^{2}\otimes\mathbb{C}^{d_E}$ be a genuinely entangled pure state. There exists a pure state $|\varphi^{(A)}\rangle\in \mathbb{C}^{2}$ on subsystem $A$, for which $\langle \varphi^{(A)}|\psi^{(ABE)}\rangle\in \mathbb{C}^{2}\otimes \mathbb{C}^{d_E}$ is an entangled vector.
\end{lem}
\begin{proof}We can write
\begin{equation}
|\psi^{(ABE)}\rangle=\alpha_1 |0^{(A)}\rangle|\psi_{1}^{(BE)}\rangle+\alpha_2 |1^{(A)}\rangle|\psi_{2}^{(BE)}\rangle
\end{equation}where $|\psi_{i}^{(BE)}\rangle\in \mathbb{C}^{2}\otimes \mathbb{C}^{d}$ for $i=1,2$ are pure states and $\alpha_1,\alpha_2\neq 0$, because if not, then $|\psi^{(ABE)}\rangle$ would be separable in $A|BE$ cut. If $|\psi_{i}^{(BE)}\rangle$ is entangled, it is enough to put $|\varphi^{(A)}\rangle=|i^{(A)}\rangle$. Assume that on the contrary for $i=1,2$ we have
\begin{equation}\label{maly}
|\psi_{i}^{(BE)}\rangle=|\phi_{i}^{(B)}\rangle|\phi_{i}^{(E)}\rangle
\end{equation}with pure states $|\phi_{i}^{(B)}\rangle \in \mathbb{C}^2, |\phi_{i}^{(E)}\rangle \in \mathbb{C}^{d_E}$, i.e. states $|\psi_{i}^{(BE)}\rangle$ are separable. Observe that vectors $|\phi_{1}^{(B)}\rangle, |\phi_{2}^{(B)}\rangle$ are linearly independent, since in other case $|\psi^{(ABE)}\rangle$ would be separable in $B|AE$ cut. In similar way entanglement $|\psi^{(ABE)}\rangle$ in $E|AB$ cut implies that vectors $|\phi_{1}^{(E)}\rangle, |\phi_{2}^{(E)}\rangle$ are linearly independent as well. Let $|\varphi^{(A)}\rangle=\beta|0^{(A)}\rangle+\gamma|0^{(A)}\rangle$ for $\beta,\gamma\neq 0$. Because of (\ref{maly}) we have
\begin{equation}\label{maly2}
\langle \varphi^{(A)}|\psi^{(ABE)}\rangle=\overline{\beta}\alpha_1|\phi_{1}^{(B)}\rangle|\phi_{1}^{(E)}\rangle+\overline{\gamma}\alpha_2|\phi_{1}^{(B)}\rangle|\phi_{1}^{(E)}\rangle.
\end{equation}
Assume that vector (\ref{maly2}) is separable. In that case there exists a vector $|\phi^{(B)}\rangle\in \mathbb{C}^2$ and a pure state $|\phi^{(E)}\rangle\in \mathbb{C}^{d_E}$ such that
\begin{equation}\label{maly3}
|\phi^{(B)}\rangle|\phi^{(E)}\rangle=\overline{\beta}\alpha_1|\phi_{1}^{(B)}\rangle|\phi_{1}^{(E)}\rangle+\overline{\gamma}\alpha_2|\phi_{2}^{(B)}\rangle|\phi_{2}^{(E)}\rangle.
\end{equation}Let $|\phi_{i,\perp}^{(B)}\rangle$ be an orthonormal to $|\phi_{i}^{(B)}\rangle$ for $i=1,2$. Therefore, we can write
\begin{equation}
\langle \phi_{1,\perp}^{(B)} |\phi^{(B)}\rangle|\phi^{(E)}\rangle=\overline{\gamma}\alpha_2\langle \phi_{1,\perp}^{(B)}|\phi_{2}^{(B)}\rangle|\phi_{2}^{(E)}\rangle
\end{equation}and
\begin{equation}
\langle \phi_{2,\perp}^{(B)} |\phi^{(B)}\rangle|\phi^{(E)}\rangle=\overline{\beta}\alpha_1\langle \phi_{2,\perp}^{(B)}|\phi_{1}^{(B)}\rangle|\phi_{1}^{(E)}\rangle,
\end{equation}where $\langle \phi_{1,\perp}^{(B)}|\phi_{2}^{(B)}\rangle, \langle \phi_{2,\perp}^{(B)}|\phi_{1}^{(B)}\rangle\neq 0$ according to linear independence of $|\phi_{1}^{(B)}\rangle, |\phi_{2}^{(B)}\rangle$ in $\mathbb{C}^2$. If so, then $|\phi_{1}^{(E)}\rangle, |\phi_{2}^{(E)}\rangle$ are linearly dependent which is a contradiction. Therefore, vector (\ref{maly2}) must be entangled.
\end{proof}

\begin{lem}\label{lemmaB}
Consider a state $|\psi^{(ABC)}\rangle\in \mathbb{C}^{2}\otimes \mathbb{C}^{2}\otimes \mathbb{C}^{d}$ entangled within the $(A|BC)$ cut. Assume that there exists a state $|\varphi^{(A)}\rangle\in \mathbb{C}^{2}$ of the first subsystem such that $\langle \varphi^{(A)}|\psi^{(ABC)}\rangle$ is entangled. Then there exists an orthonormal basis $\left\{|\varphi_i^{(A)}\rangle\right\}$ of the first subsystem such that $|\psi_i^{(BC)}\rangle=\langle \varphi_i^{(A)}|\psi^{(ABC)}\rangle$ are entangled and linearly independent.
\end{lem}
\begin{proof}Without loss of generality put $|\varphi^{(A)}\rangle=|0\rangle$ and $|\psi^{(ABC)}\rangle = \lambda_1|0\rangle(a|00\rangle + b|11\rangle) + \lambda_2|1\rangle |\phi\rangle|\phi'\rangle$ where $\lambda_1,\lambda_2, a,b\neq 0$, $|\phi\rangle = \alpha |0\rangle + \beta |1\rangle$ and $|\phi'\rangle= \alpha' |0\rangle + \beta' |1\rangle + \gamma'|\Phi_{\perp}\rangle$ with $|\Phi_{\perp}\rangle$ orthogonal to $|0\rangle$ and $|1\rangle$. Observe that due to entanglement within the $(A|BC)$ cut, $ a|00\rangle + b|11\rangle$ and $|\phi\rangle|\phi'\rangle$ are linearly independent.

Choose the orthonormal basis in which we perform the measurement $\{|\varphi_1^{(A)}\rangle = x|0\rangle + y |1\rangle, |\varphi_2^{(A)}\rangle = y|0\rangle - x|1\rangle \}$ with reals $x,y \neq 0$. Then $|\psi_1^{(BC)}\rangle = \lambda_1 x (a |00\rangle +b|11\rangle) + \lambda_2 y |\phi\rangle|\phi'\rangle$ and $|\psi_2^{(BC)}\rangle = \lambda_1 y (a |00\rangle +b|11\rangle) - \lambda_2 x |\phi\rangle|\phi'\rangle$. Note that vectors $|\psi_i^{(BC)}\rangle$ are linearly independent.

Observe that $\rho_i^{(C)}=\mathrm{Tr}_B(|\psi^{(BC)}_{i}\rangle \langle \psi^{(BC)}_{i}|)=\sum_j |\psi_{i,j}\rangle \langle \psi_{i,j}|$ where 
\begin{equation}\label{1}
|\psi_{1,0}\rangle= \lambda_1 x a |0\rangle + \lambda_2 y \alpha (\alpha' |0\rangle + \beta' |1\rangle + \gamma'|\Phi_{\perp}\rangle),
\end{equation}
\begin{equation}
\label{2}
|\psi_{1,1}\rangle= \lambda_1 x b |1\rangle + \lambda_2 y \beta (\alpha' |0\rangle + \beta' |1\rangle + \gamma'|\Phi_{\perp}\rangle),
\end{equation}
\begin{equation}\label{3}
|\psi_{2,0}\rangle=\lambda_1 y a |0\rangle - \lambda_2 x \alpha (\alpha' |0\rangle + \beta' |1\rangle + \gamma'|\Phi_{\perp}\rangle),
\end{equation}
\begin{equation}\label{4}
|\psi_{2,1}\rangle=\lambda_1  y b |1\rangle - \lambda_2 x \beta (\alpha' |0\rangle + \beta' |1\rangle + \gamma'|\Phi_{\perp}\rangle).
\end{equation}

Looking at the form of this vectors one can see that only for finite number of pairs $x,y$, vectors $|\psi_{i,j}\rangle$ can be linearly dependent (for $j=0,1$ and fixed $i$). If so, there is only finite number of pairs $x,y$ for which at least one of $\rho_i^{(C)}$ is pure (equivalently $|\psi_i^{(BC)}\rangle$ is not entangled). Therefore, we can always find some basis $\{|\varphi^{(A)}_i\rangle\}$ of the first subsystem, parametrised by $x$ and $y$, such that $|\psi^{(BC)}_i\rangle$ are entangled and linearly independent.
\end{proof}

\begin{lem}\label{lemmaC}
Consider two entangled vectors (possibly subnormalised) $|\phi_1^{(AB)}\rangle,|\phi_2^{(AB)}\rangle\in \mathbb{C}^{2}\otimes \mathbb{C}^{d}$ and a pair of non-commuting PVMs with two outcomes (described by elements $P^{(A)}_{a|x}$) on the first subsystem. For each $j=1,2$ operators $\sigma^{(B),j}_{a|x}=\mathrm{Tr}_{A}(P^{(A)}_{a|x}\otimes \mathds{1}|\phi^{(AB)}_j\rangle \langle \phi^{(AB)}_j|)$ form a row of type III. Moreover, vectors $|\phi_1^{(AB)}\rangle,|\phi_2^{(AB)}\rangle$ are linearly dependent if and only if for any choice of index $(a,x)$ operators $\sigma^{(B),1}_{a|x}, \sigma^{(B),2}_{a|x}$ are linearly dependent.
\end{lem}

\begin{proof} Without loss of generality we may put $|\phi_1^{(AB)}\rangle = \sum_i\alpha_i |i\rangle|\phi_{1,i}\rangle$ and $|\phi_2^{(AB)}\rangle = \sum_i\beta_i |i\rangle|\phi_{2,i}\rangle$ where $i=0,1$. Obviously $\alpha_i\neq 0$ (respectively $\beta_i\neq 0$) and $|\phi_{1,i}\rangle$ (respectively $|\phi_{2,i}\rangle$)  are linearly independent. The first part can be seen from direct calculations with $P^{(A)}_{0|0}=|0\rangle\langle 0|$ and $P^{(A)}_{0|1}=|\tilde{0}\rangle\langle \tilde{0}|$ where $|\tilde{0}\rangle=\gamma |0\rangle +\delta|1\rangle$ is some normalised vector with $\gamma,\delta\neq 0$. The if and only if part follows from no-signaling property of type III row expressed by (\ref{row2}).
\end{proof}

Observe that due to above Lemmas, it was established that one can choose a pair of PVMs with two outcomes on subsystems A and B respectively, in such a way that no-signaling assemblage $\Sigma^{(C)}$ (obtained by this measurements preformed on some tripartite genuine entangled pure state $|\psi^{(ABC)}\rangle$) has first column and first and second row of type III. According to the Lemma \ref{lemma1}, this particular $\Sigma^{(C)}$ is therefore inflexible and that completes the proof of Proposition 4.

\subsection{Computation of LHS and BIS bounds for Example 5}
Let us consider a three qubits GHZ state $|\psi\rangle=\frac{1}{\sqrt{2}}\left(|000\rangle +|111\rangle\right)$. Define quantum assemblage $\Sigma_{GHZ}$ according to $\sigma_{ab|xy}=\mathrm{Tr}_{AB}(P_{a|x}\otimes Q_{b|y}\otimes \mathds{1}|\psi\rangle \langle \psi|)$ where $P_{0|0}=Q_{0|0}=|+\rangle \langle +|$ and $P_{0|1}=Q_{0|1}=|0\rangle \langle 0|$, i.e. 
\begin{equation}\nonumber
\Sigma_{GHZ}=\frac{1}{4}\begin{pmatrix}
\begin{array}{cc|cc}
 |+\rangle \langle +| &  |-\rangle \langle -| & |0\rangle \langle 0| &  |1\rangle \langle 1| \\  
 |-\rangle \langle -| & |+\rangle \langle +|& |0\rangle \langle 0|  &|1\rangle \langle 1|\\ \hline
 |0\rangle \langle 0| & |0\rangle \langle 0| &2|0\rangle \langle 0|&  0  \\
   |1\rangle \langle 1| &  |1\rangle \langle 1|& 0 & 2|1\rangle \langle 1| 
\end{array}
\end{pmatrix}.
\end{equation}According to Proposition 4 $\Sigma^{(C)}_{GHZ}$ is inflexible and not biseparable. We will compute upper bounds $\mathcal{C}_{\mathbf{LHS}}$ and $\mathcal{C}_{\mathbf{BIS}}$ for inequality $F_{\Sigma^{(C)}_{GHZ}}$.
To do so we adopt the convention in which $\Sigma^{(C)}\in \bigoplus_{a,b,x,y} B(\mathbb{C}^2)_{sa}$ for any no-signaling assemblage $\Sigma^{(C)}$. With that we have $F_{\Sigma^{(C)}_{GHZ}}(\Sigma^{(C)})=\mathrm{Tr}(F_{\Sigma^{(C)}_{GHZ}}\Sigma)$ where $F_{\Sigma^{(C)}_{GHZ}}\in \bigoplus_{a,b,x,y} B(\mathbb{C}^2)_{sa}$ is of the form
\begin{equation}\label{sym_form}
F_{\Sigma^{(C)}_{GHZ}}=\begin{pmatrix}
\begin{array}{cc|cc}
 |+\rangle \langle +| &  |-\rangle \langle -| & |0\rangle \langle 0| &  |1\rangle \langle 1| \\  
 |-\rangle \langle -| & |+\rangle \langle +|& |0\rangle \langle 0|  &|1\rangle \langle 1|\\ \hline
 |0\rangle \langle 0| & |0\rangle \langle 0| &|0\rangle \langle 0|&  0  \\
   |1\rangle \langle 1| &  |1\rangle \langle 1|& 0 & |1\rangle \langle 1| 
\end{array}
\end{pmatrix}.
\end{equation}
Graphical analysis of representation (\ref{sym_form}) leads to simplification of further calculations. For further convenience we will introduce vector $\vec{\sigma}=(\sigma_x,\sigma_y,\sigma_z)$ consisting of standard Pauli matrices.

Consider an LHS assemblage $\Sigma^{(C)}_{LHS}=L\otimes |\phi\rangle \langle \phi|$ where $L$ is one of $16$ possible deterministic boxes (on susbsystem AB). Form of deterministic boxes implies that 
\begin{equation}
F_{\Sigma^{(C)}_{GHZ}}(\Sigma^{(C)}_{LHS})=\mathrm{Tr}\left((3|u\rangle \langle u|+|v\rangle \langle v|)|\phi\rangle \langle \phi|\right),
\end{equation}with $u=0,1,v=+,-$. If the above expression attains its maximal value for $(u,v)=(0,+)$ and some state $|\phi\rangle \langle \phi|$,then it attains its maximal value respectively for $(u,v)=(0,-),(1,+),(1,-)$ and states $\sigma_z|\phi\rangle \langle \phi|\sigma_z$, $\sigma_x|\phi\rangle \langle \phi|\sigma_x$ i $\sigma_x\sigma_z|\phi\rangle \langle \phi|\sigma_z\sigma_x$. In the end we may write
\begin{align}
\mathcal{C}_{\mathbf{LHS}}&=\sup_{|\phi\rangle}\mathrm{Tr}\left((3|0\rangle \langle 0|+|+\rangle \langle +|)|\phi\rangle \langle \phi|\right)\nonumber \\ \nonumber 
&=2+\sup_{|\phi\rangle}\mathrm{Tr}\left(\left(\frac{1}{2}\sigma_x+\frac{3}{2}\sigma_z\right)|\phi\rangle \langle \phi|\right)\\ 
&=\frac{4+\sqrt{10}}{2}.
\end{align}

In order to compute $\mathcal{C}_{\mathbf{BIS}}$ consider assemblages $\Sigma^{(C)}_I$ of the form (\ref{sig_1}) and $\Sigma^{(C)}_{II}$ of the form (\ref{sig_2}). As (\ref{sym_form}) is symmetric under transposition, we may write $\sup_{\Sigma^{(C)}_{I}}F_{\Sigma^{(C)}_{GHZ}}(\Sigma^{(C)}_{I})=\sup_{\Sigma^{(C)}_{II}}F_{\Sigma^{(C)}_{GHZ}}(\Sigma^{(C)}_{II})$ and restrict our attention to assemblages of the form (\ref{sig_1}) given explicitly by
\begin{widetext}
\begin{equation}\nonumber
\Sigma_{I}^{(C),1}=\begin{pmatrix}
\begin{array}{cc|cc}
 \sigma^{(C),1}_{0|0} & \sigma^{(C),1}_{1|0} & \sigma^{(C),1}_{0|1} &  \sigma^{(C),1}_{1|1} \\  
 0 & 0& 0  &0\\ \hline
 \sigma^{(C),1}_{0|0} & \sigma^{(C),1}_{1|0} &  \sigma^{(C),1}_{0|1} &   \sigma^{(C),1}_{1|1}  \\
   0 &  0& 0 & 0
\end{array}
\end{pmatrix},\ \  \Sigma_{I}^{(C),2}=\begin{pmatrix}
\begin{array}{cc|cc}
0 & 0& 0  &0\\  
\sigma^{(C),2}_{0|0} & \sigma^{(C),2}_{1|0} & \sigma^{(C),2}_{0|1} &  \sigma^{(C),2}_{1|1} \\ \hline
 \sigma^{(C),2}_{0|0} & \sigma^{(C),2}_{1|0} &  \sigma^{(C),2}_{0|1} &   \sigma^{(C),2}_{1|1}  \\
   0 &  0& 0 & 0 
\end{array}
\end{pmatrix}
\end{equation}

\begin{equation}\nonumber
\Sigma_{I}^{(C),3}=\begin{pmatrix}
\begin{array}{cc|cc}
 \sigma^{(C),3}_{0|0} & \sigma^{(C),3}_{1|0} & \sigma^{(C),3}_{0|1} &  \sigma^{(C),3}_{1|1} \\  
 0 & 0& 0  &0\\ \hline
   0 &  0& 0 & 0 \\
	\sigma^{(C),3}_{0|0} & \sigma^{(C),3}_{1|0} &  \sigma^{(C),3}_{0|1} &   \sigma^{(C),3}_{1|1}
\end{array}
\end{pmatrix},\ \  \Sigma_{I}^{(C),4}=\begin{pmatrix}
\begin{array}{cc|cc}
0 & 0& 0  &0\\  
\sigma^{(C),4}_{0|0} & \sigma^{(C),4}_{1|0} & \sigma^{(C),4}_{0|1} &  \sigma^{(C),4}_{1|1} \\ \hline
   0 &  0& 0 & 0 \\
	\sigma^{(C),4}_{0|0} & \sigma^{(C),4}_{1|0} &  \sigma^{(C),4}_{0|1} &   \sigma^{(C),4}_{1|1}
\end{array}
\end{pmatrix}
\end{equation}
\end{widetext}

Let $\Sigma_I^{(C),1}$ be described by operators $\left\{\sigma^{(C),1}_{0|0}, \sigma^{(C),1}_{1|0},\sigma^{(C),1}_{0|1},\sigma^{(C),1}_{1|1}\right\}$. If $\Sigma_I^{(C),2}$ is described by operators $\sigma^{(C),2}_{b|y}=\sigma_z\sigma^{(C),1}_{b|y}\sigma_z$ for $y,b=0,1$, then we have $F^{(C)}_{\Sigma^{(C)}_{GHZ}}(\Sigma^{(C),1}_{I})=F_{\Sigma^{(C)}_{GHZ}}(\Sigma^{(C),2}_{I})$. Similarly, if $\Sigma_I^{(C),3}$ is described by operators $\left\{\sigma^{(C),3}_{0|0}, \sigma^{(C),3}_{1|0},\sigma^{(C),3}_{0|1},\sigma^{(C),3}_{1|1}\right\}$ and $\Sigma_I^{(C),4}$ is described by operators $\sigma^{(C),4}_{0|0}=\sigma_z\sigma^{(C),3}_{b|y}\sigma_z$ for $y,b=0,1$ we have $F^{(C)}_{\Sigma_{GHZ}}(\Sigma^{(C),3}_{I})=F_{\Sigma^{(C)}_{GHZ}}(\Sigma^{(C),4}_{I})$. Finally, if $\Sigma_I^{(C),3}$ is described by operators $\sigma^{(C),3}_{b|y}=\sigma_x\sigma^{(C),1}_{b\oplus y|y}\sigma_x$, then $F_{\Sigma^{(C)}_{GHZ}}(\Sigma^{(C),1}_{I})=F_{\Sigma^{(C)}_{GHZ}}(\Sigma^{(C),3}_{I})$. Because of this reasoning, we may restrict our attention to assemblages given by $\Sigma_I^{(C),1}$.

Using representation with Pauli matrices we may express positive operators $\sigma^{(1)}_{b|y}$ by
\begin{equation}\label{pauli1}
\sigma^{(C),1}_{0|0}=\frac{\alpha_1}{2}\mathds{1}+\frac{\vec{n}\cdot\vec{\sigma}}{2},
\end{equation}
\begin{equation}\label{pauli2}
\sigma^{(C),1}_{1|0}=\frac{\alpha_2}{2}\mathds{1}+\frac{\vec{m}\cdot\vec{\sigma}}{2},
\end{equation}
\begin{equation}\label{pauli3}
\sigma^{(C),1}_{0|1}=\frac{\alpha_3}{2}\mathds{1}+\frac{\vec{l}\cdot\vec{\sigma}}{2},
\end{equation}\begin{equation}\label{pauli4}
\sigma^{(C),1}_{1|1}=\frac{\alpha_4}{2}\mathds{1}+\frac{\vec{k}\cdot\vec{\sigma}}{2},
\end{equation}where $0\leq \alpha_i \leq 1$, $|\vec{n}|\leq \alpha_1, |\vec{m}|\leq \alpha_2, |\vec{l}|\leq \alpha_3, |\vec{k}|\leq \alpha_4$ and $\alpha_1+\alpha_2=1=\alpha_3+\alpha_4$, $\vec{n}+\vec{m}=\vec{l}+\vec{k}$. Symmetries of (\ref{sym_form}) imply that $F_{\Sigma^{(C)}_{GHZ}}(\Sigma^{(C),1}_I)$ remains the same, when operators $\sigma^{(C),1}_{b|y}$ are replaced by operators $\tilde{\sigma}^{(C),1}_{b|y}=\sigma_z\sigma^{(C),1}_{b\oplus y \oplus 1|y}\sigma_z$ for $b,y=0,1$.

Without loss of generality we may then restrict only to operators satisfying $\sigma_z\sigma^{(C),1}_{0|0}\sigma_z=\sigma^{(C),1}_{1|0}$ and $\sigma_z\sigma^{(C),1}_{b|1}\sigma_z=\sigma^{(C),1}_{b|1}$ with $b=0,1$. Theretofore, (\ref{pauli1}-\ref{pauli4}) can be expressed in a simplified form
\begin{equation}\label{pauli10}
\sigma^{(C),1}_{0|0}=\frac{1}{4}\mathds{1}+\frac{n_x}{2}\sigma_x+\frac{n_y}{2}\sigma_y+\frac{n_z}{2}\sigma_z,
\end{equation}
\begin{equation}\label{pauli20}
\sigma^{(C),1}_{1|0}=\frac{1}{4}\mathds{1}-\frac{n_x}{2}\sigma_x-\frac{n_y}{2}\sigma_y+\frac{n_z}{2}\sigma_z,
\end{equation}
\begin{equation}\label{pauli30}
\sigma^{(C),1}_{0|1}=\frac{\alpha_3}{2}\mathds{1}+\frac{l_z}{2}\sigma_z,
\end{equation}\begin{equation}\label{pauli40}
\sigma^{(C),1}_{1|1}=\frac{\alpha_4}{2}\mathds{1}+\frac{k_z}{2}\sigma_z,
\end{equation}with $0\leq \alpha_3 \leq 1$, $\alpha_4=1-\alpha_3$, $|\vec{n}|\leq \frac{1}{2}$, $|l_z|\leq \alpha_3, |k_z|\leq 1-\alpha_3$ and $2n_z=l_z+k_z$.
According to (\ref{pauli10}-\ref{pauli40}) we may write
\begin{equation}\label{ost_szac}
F_{\Sigma^{(C)}_{GHZ}}(\Sigma^{(C),1}_I)=\frac{3}{2}+n_x+\frac{1}{2}\alpha_3+\frac{3}{2}l_z.
\end{equation}

Without loss of generality we may moreover assume that $n_x\geq 0$ and $n_y=0$, $|\vec{n}|= \frac{1}{2}$ and put $n_x=\sqrt{\frac{1}{4}-n_z^2}$. Expression (\ref{ost_szac}) reduce in this case to
\begin{equation}\label{ost_szac2rrrrrr}
f(n_z,\alpha_3,l_z)=\frac{3}{2}+\sqrt{\frac{1}{4}-n_z^2}+\frac{1}{2}\alpha_3+\frac{3}{2}l_z,
\end{equation}under conditions $0\leq \alpha_3\leq 1$, $|l_z|\leq \alpha_3$ and $|2n_z-l_z|\leq 1-\alpha_3$. Using explicitly Kuhn-Tucker method \cite{kuhn_tucker} or Mathematica package \cite{mathematica} it can be shown that (under constraints given above)
\begin{equation}
\sup_{\Sigma^{(C)}_I}F_{\Sigma^{(C)}_{GHZ}}(\Sigma^{(C)}_I)=\sup_{n_z,\alpha_3,l_z}f(n_z,\alpha_3,l_z)=\frac{5+\sqrt{5}}{2}.
\end{equation}

Let as now consider any assemblage of the form $\Sigma^{(C)}_{III}=P\otimes |\phi\rangle \langle \phi|$ where $P$ is some no-signaling box (assemblages of this form provide a relaxation of (\ref{sig_3}). 
\begin{equation}\label{P1}
P=\begin{pmatrix}
\begin{array}{cc|cc}
 p_1 &  p_2 & q_1 &  q_2 \\  
 p_3 & p_4 & q_3 &q_4\\ \hline
 t_1 & t_2 &  r_1&  r_2 \\
  t_3 & t_4& r_3 & r_4 
\end{array}
\end{pmatrix}.
\end{equation}
Symmetry of (\ref{sym_form}) with respect to transposition implies that the value of $F_{\Sigma_{GHZ}}(\Sigma_{III})$ remains the same if box $P$ of the form (\ref{P1}) is replaced by
\begin{equation}\label{P2}
P=\begin{pmatrix}
\begin{array}{cc|cc}
 p_1 &  p_3 & t_1 &  t_3 \\  
 p_2 & p_4 & t_2 &t_4\\ \hline
 q_1 & q_3 &  r_1&  r_3 \\
  q_2 & q_4& r_2 & r_4 
\end{array}
\end{pmatrix}.
\end{equation}If so, then we can restrict our attention to $P$ given by
\begin{equation}\label{P3}
P=\begin{pmatrix}
\begin{array}{cc|cc}
 p_1 &  p_2 & q_1 &  q_2 \\  
 p_2 & p_4 & q_3 &q_4\\ \hline
 q_1 & q_3 &  r_1&  r_2 \\
  q_2 & q_4& r_2 & r_4 
\end{array}
\end{pmatrix}.
\end{equation}Observe that (\ref{sym_form}) is invariant with respect to exchange between first and second rows composed with exchange between first and second columns. Therefore, the value of $F_{\Sigma^{(C)}_{GHZ}}(\Sigma^{(C)}_{III})$ remains the same if $P$ given by (\ref{P3}) is replace by the no-signaling box of the following form
\begin{equation}\label{P4}
P=\begin{pmatrix}
\begin{array}{cc|cc}
 p_4 &  p_2 & q_3 &  q_4 \\  
 p_2 & p_1 & q_1 &q_2\\ \hline
 q_3 & q_1 &  r_1&  r_2 \\
  q_4 & q_2& r_2 & r_4 
\end{array}
\end{pmatrix}.
\end{equation}Without loss of generality it is enough to consider $P$ given as
\begin{equation}\label{P5}
P=\begin{pmatrix}
\begin{array}{cc|cc}
 p_1 &  p_2 & q_1 &  q_2 \\  
 p_2 & p_1 & q_1 &q_2\\ \hline
 q_1 & q_1 &  r_1&  r_2 \\
  q_2 & q_2& r_2 & r_4 
\end{array}
\end{pmatrix}.
\end{equation}Observe that the value of $F_{\Sigma^{(C)}_{GHZ}}(\Sigma_{III})$ cannot decrease if a no-signaling box $P$ of the form (\ref{P5}) is replaced by
\begin{equation}\label{P6}
P=\begin{pmatrix}
\begin{array}{cc|cc}
 p_1 &  p_2 & q_1 &  q_2 \\  
 p_2 & p_1 & q_1 &q_2\\ \hline
 q_1 & q_1 &  r'_1&  0 \\
  q_2 & q_2& 0 & r'_4 
\end{array}
\end{pmatrix}
\end{equation}with $r'_1=r_1+r_2$ and $r'_4=r_4+r_2$. Therefore, it is enough to consider no-signaling boxes $P$ parametrized by
\begin{equation}
P=\begin{pmatrix}
\begin{array}{cc|cc}
 p &  \frac{1}{2}-p & \frac{r}{2} &  \frac{1-r}{2} \\  
 \frac{1}{2}-p & p & \frac{r}{2} &\frac{1-r}{2}\\ \hline
 \frac{r}{2} & \frac{r}{2} &  r&  0 \\
   \frac{1-r}{2} &  \frac{1-r}{2}& 0 & 1-r 
\end{array}
\end{pmatrix}
\end{equation}where $0\leq p\leq \frac{1}{2}$ and $0\leq r\leq 1$. With this parametrization we obtain
\begin{widetext}
\begin{equation}
F_{\Sigma^{(C)}_{GHZ}}(\Sigma^{(C)}_{III})=\mathrm{Tr}\left(|\phi\rangle \langle \phi|\left(2p|+\rangle \langle +|+(1-2p)|-\rangle \langle -|+3r|0\rangle \langle 0|+3(1-r)|1\rangle \langle 1|\right)\right)
\end{equation}which implies
\begin{align}
\sup_{\Sigma^{(C)}_{III}}F_{\Sigma^{(C)}_{GHZ}}(\Sigma_{III})&=2+\sup_{p,r}\left(\sup_{|\phi\rangle} \mathrm{Tr}\left(|\phi\rangle \langle \phi|\left(\left(2p-\frac{1}{2}\right)\sigma_x+\frac{3}{2}\left(2r-1\right)\sigma_z\right)\right)\right) \nonumber \\
&=2+\sup_{p,r}\sqrt{\left(2p-\frac{1}{2}\right)^2+\frac{9}{4}\left(2r-1\right)^2}=\frac{4+\sqrt{10}}{2}=\mathcal{C}_{\mathbf{LHS}},
\end{align}
\end{widetext}
so initial relaxation of (\ref{sig_3}) was admissible.

Finally, we get
\begin{equation}
\mathcal{C}_{\mathbf{BIS}}=\max_{i=I,II,III}\left\{\sup_{\Sigma^{(C)}_{i}}F_{\Sigma^{(C)}_{GHZ}}(\Sigma^{(C)}_{i})\right\}=\frac{5+\sqrt{5}}{2}
\end{equation}and $\mathcal{C}_{\mathbf{LHS}}<\mathcal{C}_{\mathbf{BIS}}$.

\subsection{Proof of Proposition 6 - application of Jordan's Lemma}

To conclude the proof of nontrivial implication of Proposition 6, recall the following Lemma.

\begin{lem}[Jordan's Lemma \cite{Jordan}]\label{lemma5} Let $P,Q\in M_n(\mathbb{C})$ be two projections. Then the space $\mathbb{C}^n=\bigoplus_i H_i$ can be decompose as a direct sum of subspaces $H_i$ of dimension lesser than or equal to two in such a way that projections $P$ and $Q$ restricted to subspaces $H_i$ are either given by
\begin{equation}\nonumber
P,Q|_{H_i}\in \left\{(0),(1)\right\}
\end{equation}when $\mathrm{dim}\ H_i =1$ or 
\begin{equation}\nonumber
P|_{H_i}=\begin{pmatrix}
    1& 0  \\
    0& 0 
\end{pmatrix},\   Q|_{H_i}=\begin{pmatrix}
    \cos^2\theta_i & \cos\theta_i\sin\theta_i  \\
    \cos\theta_i\sin\theta_i &\sin^2\theta_i 
\end{pmatrix}
\end{equation}for some angle $\theta_i\in \left[0,\frac{\pi}{2}\right)$, when $\mathrm{dim}\ H_i =2$, i.e. for any two projections $P,Q$ there exists an orthonormal basis for witch $P$ and $Q$ are simultaneously block diagonal.
\end{lem}

Consider now any pure state $|\tilde{\psi}^{(A'B'C)}\rangle\in \mathbb{C}^{d_{A'}}\otimes \mathbb{C}^{d_{B'}}\otimes \mathbb{C}^{d_{C}}$ and assemblage $\tilde{\Sigma}^{(C)}$ given by two pairs of PVMs with two outcomes $\tilde{\sigma}^{(C)}_{ab|xy}=\mathrm{Tr}_{A'B'}(\tilde{P}^{(A')}_{a|x}\otimes \tilde{Q}^{(B')}_{b|y}\otimes \mathds{1}|\tilde{\psi}^{(A'B'C)}\rangle \langle \tilde{\psi}^{(A'B'C)}|)$. By Jordan's Lemma two pairs of projections $\tilde{P}^{(A')}_{0|0}, \tilde{P}^{(A')}_{0|1}$ on subsystem A' and $\tilde{Q}^{(B')}_{0|0}, \tilde{Q}^{(B')}_{0|1}$ on subsystem B' can be considered as block-diagonal operators acting on a direct sum of qubit spaces (up to local isometries putting $\mathbb{C}$ inside $\mathbb{C}^2$). Therefore, effectively assemblage $\tilde{\Sigma}^{(C)}$ is given by a convex combination of assemblages obtained by PVMs performed on states of $\mathbb{C}^2\otimes \mathbb{C}^2\otimes \mathbb{C}^{d_{C}}$. If $F_{\Sigma^{(C)}}(\tilde{\Sigma}^{(C)})=4$ then each of this states must be up to local unitaries equal to $|\psi^{(ABC)}\rangle$ (which defines $F_{\Sigma^{(C)}}$). Using isomorphism $\bigoplus_i^k\mathbb{C}^2=\mathbb{C}^k\otimes \mathbb{C}^2$ for subsystems A' and B' we conclude the proof.

\subsection{Example of infexibility in tripartite scenario with $d$-outcomes}

In this subsection, we will provide an additional example of a construction of iflexible (hence extremal and exposed) no-signaling assemblage with quantum realisation. This construction in a tripartite scenario where two parties perform two measurements, each with $d$ outcomes, is not based on sufficient condition of inflexibility given in a central Lemma \ref{lemma1} from Subsection \ref{sec}. 

Consider an assemblage $\Sigma^{(C)}_{GHZ_d}=\left\{\sigma^{(C)}_{ab|xy}\right\}_{a,b,x,y}$ define by
\begin{equation}
\sigma^{(C)}_{ab|xy}=\mathrm{Tr}_{AB}(P^{(A)}_{a|x}\otimes Q^{(B)}_{b|y}\otimes \mathds{1}|GHZ^{(ABC)}_d\rangle \langle GHZ^{(ABC)}_d|)
\end{equation}with $P^{(A)}_{k|0}=|k\rangle \langle k|$, $P^{(A)}_{k|1}=|\phi_k\rangle \langle \phi_k|$, $P^{(A)}_{a|x}=Q^{(B)}_{a|x}$ and $a,b=0,\ldots, d-1$, $x,y=0,1$, where 
\begin{equation}
|GHZ^{(ABC)}_d\rangle=\frac{1}{\sqrt{d}}\sum_{k=0}^{d-1}|kkk\rangle.
\end{equation}and 
\begin{equation}
|\phi_k\rangle=\frac{1}{\sqrt{d}}\sum_{l=0}^{d-1}\omega^{lk}|l\rangle 
\end{equation}with $\omega=e^{\frac{2\pi i}{d}}$ i $k=0,\ldots, d-1$. Straightforward computations show that
\begin{eqnarray}\label{ghz1-4}
\sigma^{(C)}_{ab|00}=\frac{1}{d}\delta_{ab}|a\rangle \langle a|, \nonumber \\
\sigma^{(C)}_{ab|01}=\frac{1}{d^2}|a\rangle \langle a|, \nonumber \\
\sigma^{(C)}_{ab|10}=\frac{1}{d^2}|b\rangle \langle b|, \nonumber \\
\sigma^{(C)}_{ab|11}=\frac{1}{d^2}|\phi_r\rangle \langle \phi_r|,
\end{eqnarray}
with $r=d-a-b\ \mathrm{mod}\ d$.


To see that $\Sigma^{(ABC)}_{GHZ_d}$ is inflexible consider
\begin{equation}\label{equality}
\sum_{k=0}^{d-1} |\tilde{k}\rangle \langle \tilde{k}|=\sum_{k=0}^{d-1} |\tilde{\phi}_k\rangle \langle \tilde{\phi}_k|,
\end{equation}with $|\tilde{k}\rangle =\frac{1}{d}|k\rangle$ and $|\tilde{\phi}_k\rangle =\frac{1}{d}|\phi_k\rangle$ as well as another equality
\begin{equation}\label{test}
\sum_{k=0}^{d-1}\alpha_k |\tilde{k}\rangle \langle \tilde{k}|=\sum_{k=0}^{d-1} \beta_k|\tilde{\phi}_k\rangle \langle \tilde{\phi}_k|
\end{equation}with some $\alpha_k,\beta_k\geq 0$. Observe that if equality \ref{test} holds then either $\alpha_k,\beta_k=0$ for all $k$ or $\alpha_k,\beta_k>0$ for all $k$. We will show that moreover it implies that $\alpha_k=\beta_k=\alpha$ for some constant $\alpha$. 

Indeed, there exists a unitary matrix $U$, such that
\begin{equation}\label{eq1}
|\tilde{\phi}_k\rangle=\sum_{i=0}^{d-1}u_{ki}|\tilde{i}\rangle
\end{equation}and $|u_{ik}|=\frac{1}{\sqrt{d}}$. According to \cite{HJW93} and (\ref{test}) there exists another unitary $\tilde{U}$ such that
\begin{equation}\label{eq2}
\sqrt{\beta_k}|\tilde{\phi}_k\rangle=\sum_{i=0}^{d-1}\tilde{u}_{ki} \sqrt{\alpha_i}|\tilde{i}\rangle.
\end{equation}If so, then
\begin{equation}\label{r}
\sqrt{\beta_k}u_{ki}=\tilde{u}_{ki} \sqrt{\alpha_i}
\end{equation}for all $k,i$.

Without loss of generality, set $\alpha_1=\max_i \left\{\alpha_i\right\}$, $\alpha_2=\min_i \left\{\alpha_i\right\}$ and assume $\alpha_1>\alpha_2$. As for a fixed $k$ we have $|\tilde{u}_{ki}| \sqrt{\alpha_i}=|\tilde{u}_{kj}| \sqrt{\alpha_j}$ for all $i,j$, the above assumption implies that $|\tilde{u}_{k1}|<|\tilde{u}_{k2}|$. On the other hand, 
\begin{equation}
\sum_{k=0}^{d-1}|\tilde{u}_{k1}|^2=1=\sum_{k=0}^{d-1}|\tilde{u}_{k2}|^2
\end{equation}which is a contradiction. If so, then $\alpha_i=\alpha$ for all $i$ and since 
\begin{equation}
\sum_{i=0}^{d-1}\beta_k|u_{ki}|^2=\sum_{i=0}^{d-1}\alpha|\tilde{u}_{ki}|^2
\end{equation}unitarity of $U$ and $\tilde{U}$ implies $\beta_k=\alpha$ for any $k$.

 If $\Sigma^{(C)}$ is similar to $\Sigma^{(C)}_{GHZ_d}$, then obviously $\Sigma=\left\{\alpha_{ab|xy}\sigma^{(C)}_{ab|xy}\right\}_{a,b,x,y}$ with $\alpha_{ab|xy}\geq 0$. Moreover, for $y=1$ and $b$ 
\begin{equation}\label{r1}
\sum_a\alpha^{(C)}_{ab|11}\sigma_{ab|11}=\sum_a\alpha_{ab|01}\sigma^{(C)}_{ab|01},
\end{equation}while for $x=1$ and $a$
\begin{equation}\label{r2}
\sum_b\alpha_{ab|11}\sigma^{(C)}_{ab|11}=\sum_b\alpha_{ab|10}\sigma^{(C)}_{ab|10}.
\end{equation}Observe that analogous no-signaling constrains for $\Sigma^{(C)}_{GHZ_d}$ are given by (\ref{equality}). Therefore, by (\ref{r1}-\ref{r2}), (\ref{ghz1-4}) and the previous reasoning and (\ref{r1}-\ref{r2}), (\ref{ghz1-4}) we have $\alpha_{ab|xy}=1$ for all $a,b,x,y$ (because of normalization). Finally, $\Sigma^{(C)}=\Sigma^{(C)}_{GHZ_d}$ and $\Sigma^{(C)}_{GHZ_d}$ form an example of inflexible (hence extremal and exposed) assemblage. 

Note that according to (\ref{sig_1}), (\ref{sig_2}) i (\ref{sig_3}) shows that considered assemblage $\Sigma^{(C)}_{GHZ_d}$ is not biseparable (hence does nit admit an LHS model), therefore construction of functional $F_{\Sigma^{(C)}_{GHZ_d}}$ (see Subsection \ref{multipartite_case}) provides an example of steering inequality detecting tripartite genuine entanglement (including tripartite GHZ state $|GHZ^{(ABC)}_d\rangle$).

\subsection{Remarks on quantum realisations of no-signaling assemblages in different scenarios}

The main question of our interest asks whether a quantum assemblage can be an extremal non-classical point in the larger convex set of all no-signaling assemblages. In the case of bipartite steering, each no-signaling assemblage admits a quantum realisation, therefore such a question has a trivial answer or in other words it is somehow ill-defined. Below we identify the minimal scenario in which the question has a nontrivial answer. The simplest generalisation of the bipartite steering scenario is a tripartite case in which two uncharacterised parties (A, B) steer a characterised one (C). Firstly assume that Alice and Bob may perform local (and independent) measurements with settings labeled by $x$ or $y$ (respectively) and outcomes labeled by $a$ or $b$ (respectively), where $a,b,y\in\left\{0,1\right\}$ and $x=0$, i.e. there is \textit{only one measurement setting} for subsystem A. This situation is described by a no-signaling assemblage $\Sigma^{(C)}=\left\{\sigma^{(C)}_{ab|0y}\right\}_{a,b,y}$ with box representation (\ref{sup_box})
\begin{equation}
\Sigma^{(C)}=\begin{pmatrix}
\begin{array}{cc|cc}
 \sigma^{(C)}_{00|00} &  \sigma^{(C)}_{01|00} & \sigma^{(C)}_{00|01} &  \sigma^{(C)}_{01|01} \\  
 \sigma^{(C)}_{10|00} & \sigma^{(C)}_{11|00}& \sigma^{(C)}_{10|01}  & \sigma^{(C)}_{11|01}
\end{array}
\end{pmatrix}.
\end{equation}Take any PVM $P^{(A)}_{a|0}$ on subsystem A. Define bipartite assemblage $\tilde{\Sigma}^{(AC)}=\left\{\tilde{\sigma}^{(AC)}_{b|y}\right\}_{b,y}$ where $\tilde{\sigma}^{(AC)}_{b|y}=P^{(A)}_{0|0}\otimes \sigma^{(C)}_{0b|0y}+P^{(A)}_{1|0}\otimes\sigma^{(C)}_{1b|0y}$. Because any bipartite no-signaling assemblage admits a quantum realisation, there exist a subsystem B, POVMs given by $N^{(B)}_{b|y}$ and a state $\rho^{(ABC)}$ such that
\begin{equation}
\tilde{\sigma}^{(AC)}_{b|y}=\mathrm{Tr}_{B}(N^{(B)}_{b|y}\otimes \mathds{1}\rho^{(ABC)})
\end{equation}and as a consequence 
\begin{equation}
\sigma^{(C)}_{ab|0y}=\mathrm{Tr}_{AB}(P^{(A)}_{a|0}\otimes N^{(B)}_{b|y}\otimes \mathds{1}\rho^{(ABC)}).
\end{equation}Note that above reasoning still holds if we exchange the roles of Alice and Bob. To summarize any no-signaling assemblage coming form  tripartie scenario with only one measurement setting in either of uncharacterised parties still admits quantum realisation. The situation when both uncharacterised parties have a single measurement setting is obviously trivial. The whole above analysis shows that the steering scenario in which two uncharacterised parties steer to characterised one by measurements described respectively by labels of settings and outcomes $a,b,x,y\in \left\{0,1\right\}$, is the first nontrivial setting in which one can ask the initial question.

Since we wish to discuss in this setting quantum realisation of extremal assemblages, according to Section \ref{simp}, we may restrict our attention only to tripartite states which are pure. Assume for the moment that $|\psi^{(ABC)}\rangle=|\phi^{(AB)}\rangle|\phi^{(C)}\rangle$ is separable in the cut $AB|C$. For any choice of POVMs given by $M^{(A)}_{a|x}, N^{(B)}_{b|y}$ consider
\begin{align}\nonumber
\sigma_{ab|xy}^{(C)}&=\mathrm{Tr}_{AB}(M^{(A)}_{a|x}\otimes N^{(B)}_{b|y}\otimes \mathds{1}|\phi^{(AB)}\rangle|\phi^{(C)}\rangle \langle \phi^{(AB)}|\langle\phi^{(C)}|)\\ \nonumber
&=p^{(AB)}(ab|xy)|\phi^{(C)}\rangle \langle\phi^{(C)}|. \nonumber
\end{align}Assemblage defined like that is extremal if and only if $p^{(AB)}(ab|xy)$ is extremal in a polytope of no-signaling correlations. Since this box of probabilities admits quantum realisation, the former is possible only if $p^{(AB)}(ab|xy)$ describes one of local boxes, but in this case extremal assemblage is LHS and hence trivial.

To consider other possibility, assume without loss of generality that $|\psi^{(ABC)}\rangle=|\phi^{(A)}\rangle|\phi^{(BC)}\rangle$ is separable in the cut $A|BC$ and $|\phi^{(BC)}\rangle$ is entangled. For any choice of POVMs given by $M^{(A)}_{a|x}, N^{(B)}_{b|y}$ consider 
\begin{align}\nonumber
\sigma_{ab|xy}^{(C)}&=\mathrm{Tr}_{AB}(M^{(A)}_{a|x}\otimes N^{(B)}_{b|y}\otimes \mathds{1}|\phi^{(A)}\rangle |\phi^{(BC)}\rangle \langle \phi^{(A)}| \langle\phi^{(BC)}|)\\ \nonumber
&=\mathrm{Tr}_{A}(M^{(A)}_{a|x}|\phi^{(A)}\rangle\langle \phi^{(A)}|)\sigma_{b|y}^{(C)}.\nonumber
\end{align}where $\sigma_{b|y}^{(C)}=\mathrm{Tr}_{B}(N^{(B)}_{b|y}\otimes \mathds{1}|\phi^{(BC)}\rangle\langle\phi^{(BC)}|)$ form a bipartite no-signaling assemblage. If we put $\alpha_{a|x}=\mathrm{Tr}_{A}(M^{(A)}_{a|x}|\phi^{(A)}\rangle\langle \phi^{(A)}|)$ we may represents this assemblage in box form (\ref{sup_box}) as
\begin{equation}\nonumber
\Sigma^{(C)}=\begin{pmatrix}
\begin{array}{cc|cc}
 \alpha_{0|0}\sigma^{(C)}_{0|0} &  \alpha_{0|0}\sigma^{(C)}_{1|0} & \alpha_{0|0}\sigma^{(C)}_{0|1} &  \alpha_{0|0}\sigma^{(C)}_{1|1} \\  
 \alpha_{1|0}\sigma^{(C)}_{0|0} & \alpha_{1|0}\sigma^{(C)}_{1|0}& \alpha_{1|0}\sigma^{(C)}_{0|1}  & \alpha_{1|0}\sigma^{(C)}_{1|1} \\ \hline
 \alpha_{0|1}\sigma^{(C)}_{0|0} & \alpha_{0|1}\sigma^{(C)}_{1|0} & \alpha_{0|1}\sigma^{(C)}_{0|1}&  \alpha_{0|1}\sigma^{(C)}_{1|1}  \\
   \alpha_{1|1}\sigma^{(C)}_{0|0} & \alpha_{1|1}\sigma^{(C)}_{1|0} & \alpha_{1|1}\sigma^{(C)}_{0|1} & \alpha_{1|1}\sigma^{(C)}_{1|1}
\end{array}
\end{pmatrix}.
\end{equation}If both $\alpha_{0|0}$ and $\alpha_{1|0}$ (respectively $\alpha_{0|1}$ and $\alpha_{1|1}$) are nonzero, then $\Sigma^{(C)}$ is not extremal. Indeed, $\Sigma^{(C)}=\alpha_{0|0}\Sigma^{(C)}_{0|0}+\alpha_{1|0}\Sigma^{(C)}_{1|0}$ where
\begin{equation}\nonumber
\Sigma^{(C)}_{0|0}=\begin{pmatrix}
\begin{array}{cc|cc}
 \sigma^{(C)}_{0|0} &  \sigma^{(C)}_{1|0} & \sigma^{(C)}_{0|1} &  \sigma^{(C)}_{1|1} \\  
 0 & 0& 0  & 0 \\ \hline
 \alpha_{0|1}\sigma^{(C)}_{0|0} & \alpha_{0|1}\sigma^{(C)}_{1|0} & \alpha_{0|1}\sigma^{(C)}_{0|1}&  \alpha_{0|1}\sigma^{(C)}_{1|1}  \\
   \alpha_{1|1}\sigma^{(C)}_{0|0} & \alpha_{1|1}\sigma^{(C)}_{1|0} & \alpha_{1|1}\sigma^{(C)}_{0|1} & \alpha_{1|1}\sigma^{(C)}_{1|1}
\end{array}
\end{pmatrix}
\end{equation}and
\begin{equation}\nonumber
\Sigma^{(C)}_{1|0}=\begin{pmatrix}
\begin{array}{cc|cc}
 0 &  0 & 0 &  0 \\  
 \sigma^{(C)}_{0|0} &\sigma^{(C)}_{1|0}& \sigma^{(C)}_{0|1}  & \sigma^{(C)}_{1|1} \\ \hline
 \alpha_{0|1}\sigma^{(C)}_{0|0} & \alpha_{0|1}\sigma^{(C)}_{1|0} & \alpha_{0|1}\sigma^{(C)}_{0|1}&  \alpha_{0|1}\sigma^{(C)}_{1|1}  \\
   \alpha_{1|1}\sigma^{(C)}_{0|0} & \alpha_{1|1}\sigma^{(C)}_{1|0} & \alpha_{1|1}\sigma^{(C)}_{0|1} & \alpha_{1|1}\sigma^{(C)}_{1|1}
\end{array}
\end{pmatrix}.
\end{equation}The only way to obtain extremal assemblage is to, without loss of generality, put $M^{(A)}_{0|x} =|\phi^{(A)}\rangle \langle \phi^{(A)}|$ for $x=0,1$ and obtain 
\begin{equation}\nonumber
\Sigma^{(C)}=\begin{pmatrix}
\begin{array}{cc|cc}
 \sigma^{(C)}_{0|0} &  \sigma^{(C)}_{1|0} & \sigma^{(C)}_{0|1} & \sigma^{(C)}_{1|1} \\  
 0 & 0 & 0  & 0 \\ \hline
   \sigma^{(C)}_{0|0} & \sigma^{(C)}_{1|0} & \sigma^{(C)}_{0|1} & \sigma^{(C)}_{1|1}\\
	 0 & 0 & 0 &  0 
\end{array}
\end{pmatrix}.
\end{equation}If one choose $N^{(B)}_{b|y}$ that $\left\{\sigma^{(C)}_{b|y}\right\}_{b,y}$ form a row of type III (it is always possible due to entanglement of $|\phi^{(BC)}\rangle$ and Lemma \ref{lemmaC}), then, according to discussion in Section \ref{sec}, $\Sigma^{(C)}$ will be indeed extremal (and without LHS model).

However, extremality of that assemblage follows only from the lack of distinction between the set no-signaling assemblages and quantum assemblages in a bipartite case. This can be seen simply as an inclusion of bipartite scenario into a tripartite one. As examples of such type are somehow trivial (despite being not LHS), it explains the focus of Proposition 4 on the case of a genuine entangled pure tripartite states $|\psi^{(ABC)}\rangle$, which provide results that cannot be reduced to bipartite situation.

\end{document}